\documentclass[11pt]{amsart}
\usepackage{amsmath}
\usepackage[dvips]{epsfig}
\usepackage{slashed}
\usepackage{accents}
\usepackage{color}
\usepackage{leftidx}

\theoremstyle{plain}
\newtheorem{theorem}{Theorem}[section]
\newtheorem{lemma}[theorem]{Lemma}

\newtheorem{proposition}[theorem]{Proposition}

\theoremstyle{definition}
\newtheorem{remark}[theorem]{Remark}

\numberwithin{equation}{section}
\newcommand{\volg}{1+ u^{2}| \nabla f |_{g}^{2}} %volume form of \bg%
\newcommand{\barna}{\overline{\nabla}} %barred nabla - covariant derivative w.r.t. barred g%
\begin{document}

\title[Transformations of Asymptotically AdS Hyperbolic Initial Data] {Transformations of Asymptotically AdS Hyperbolic Initial Data and Associated Geometric Inequalities}

\author[Cha]{Ye Sle Cha}
\address{Institut f\"ur Mathematik \\
 Freie Universit\"at Berlin \\
 14195 Berlin, Germany}
\email{ycha@zedat.fu-berlin.de}

\author[Khuri]{Marcus Khuri}
\address{Department of Mathematics\\
Stony Brook University\\
Stony Brook, NY 11794, USA}
\email{khuri@math.sunysb.edu}

\thanks{Y. Cha acknowledges the support of
Dahlem Research School Fellowship from Free University of Berlin. M. Khuri acknowledges the support of NSF Grants DMS-1308753 and DMS-1708798.}

\begin{abstract}
We construct transformations which take asymptotically AdS hyperbolic initial data into asymptotically flat initial data, and which preserve relevant physical quantities. This is used to derive geometric inequalities in the asymptotically AdS hyperbolic setting from counterparts in the asymptotically flat realm,
whenever a geometrically motivated system of elliptic equations admits a solution. The inequalities treated here relate mass, angular momentum, charge, and horizon area.
\end{abstract}
\maketitle

\section{Introduction}
\label{sec1} \setcounter{equation}{0}
\setcounter{section}{1}

Geometric inequalities relating total mass, angular momentum, charge, and horizon area have been studied extensively in the setting of asymptotically flat initial data for the Einstein-Maxwell equations. These inequalities may be stated without alteration for asymptotically hyperboloidal slices arising from asymptotically flat spacetimes. In \cite{ChaKhuriSakovich} we proposed deformations of asymptotically hyperboloidal initial data which transformed the asymptotically hyperbolic structure into an asymptotically flat structure, while preserving the mass up to scaling by a positive constant. These deformations are based on solutions of the so called Jang-type equations \cite{BrayKhuri1,BrayKhuri2,ChaKhuri1,ChaKhuri2,Jang,SchoenYau}, which impart positivity properties to the scalar curvature of the deformed data. In effect, the full round of geometric inequalities in the asymptotically hyperboloidal setting is reduced to their counterparts in the asymptotically flat setting modulo the solution of canonical system of partial differential equations.

Asymptotically hyperbolic data also arise naturally in the context of spacelike slices of asymptotically anti-de Sitter (AdS) spacetimes. In this realm, very little is known about the geometric inequalities discussed above, and in some cases the optimal statement of the inequalities has not been discussed in the literature. Unlike in the asymptotically flat and asymptotically hyperboloidal cases, where Penrose's heuristic arguments \cite{Penrose} based on cosmic censorship and the final state conjecture are valid, the physical motivation for the AdS inequalities is weak as a consequence of the likely instability of the AdS spacetime \cite{BizonRostworowski,Moschidis} and its black hole configurations. Furthermore, the inequalities which have been proven in the AdS setting have essentially all been established via spinorial techniques, modeled on Witten's proof of the positive mass theorem \cite{Witten}. In this article an alternate approach will be provided. Namely, with the goal of generalizing \cite{ChaKhuriSakovich} to the AdS setting, we construct deformation procedures which reduce a range of geometric inequalities to solving a canonical system of equations. The strategy is to transform the asymptotically AdS structure to an asymptotically flat structure, while preserving relevant quantities and achieving nonnegative scalar curvature at least weakly. One may then derive results in the AdS context from those in the asymptotically flat realm. Some of the inequalities that naturally arise from this method are new to the AdS setting, and for some values of the angular momentum and charge are stronger than those previously considered.
We also formulate the full set of optimal inequalities in the AdS context, filling a gap in the literature.

\section{Notation and Definitions}
\label{sec2} \setcounter{equation}{0}
\setcounter{section}{2}

Consider an initial data set $(M,g,k)$ for the Einstein equations modeling an asymptotically totally geodesic, spacelike hypersurface, in an asymptotically AdS spacetime. This consists of a Riemannian 3-manifold $M$ with asymptotically hyperbolic metric $g$, and asymptotically vanishing extrinsic curvature $k$. We define such an initial data set to be asymptotically AdS hyperbolic if it possesses an end which is diffeomorphic to $S^{2}\times [r_{0},\infty)$, and in this region there are coordinates such that
\begin{equation}\label{1}
g=g_{0}+a,
\end{equation}
where $g_{0}=\frac{dr^{2}}{1+r^{2}}+r^{2}\sigma$ is the hyperbolic metric with $\sigma$ the round metric on $S^2$, and
\begin{equation}\label{2}
a_{rr}= \frac{\mathbf{m}^{r}}{r^{5}}+O_{3}(r^{-6}), \text{ }\text{ }\text{ }\text{ }\text{ }\text{ }
a_{r\alpha}=O_{3}(r^{-3}), \text{ }\text{ }\text{ }\text{ }\text{ }\text{ } a_{\alpha\beta}=\frac{\mathbf{m}^{g}_{\alpha\beta}}{r}+ O_{3}(r^{-2}),
\end{equation}
%\begin{equation}\label{2.1}
%\widetilde{a}_{\alpha\beta},\widetilde{b}_{\alpha\beta}=O(r^{-2}),\text{ }\text{ }\text{ %}\text{ }\text{ }\partial_{\gamma}\widetilde{a}_{\alpha\beta}=O(r^{-2}),\text{ }\text{ %}\text{ }\text{ }\text{ }\partial_{r}\widetilde{a}_{\alpha\beta}=O(r^{-3}),
%\end{equation}
\begin{equation}\label{3}
k_{rr}=O_{2}(r^{-5}), \text{ }\text{ }\text{ }\text{ }\text{ }\text{ }  k_{r\alpha}=O_{2}(r^{-3}),\text{ }\text{ }\text{ }\text{ }\text{ }\text{ } k_{\alpha\beta}= \frac{\mathbf{m}^{k}_{\alpha\beta}}{r} + O_{2}(r^{-2}),
\end{equation}
with Greek letters denoting indices for coordinates on $S^{2}$.
Here $\mathbf{m}^{g}$ and $\mathbf{m}^{k}$ are tensors and $\mathbf{m}^{r}$ is a function, all on $S^{2}$ and independent of $r$.

The quantities $\mathbf{m}^{g}$ and $\mathbf{m}^{r}$ encode mass through the formula
\begin{equation}\label{3.0}
m=\frac{1}{16\pi}\int_{S^{2}} 3Tr_{\sigma}\mathbf{m}^{g} + 2\mathbf{m}^{r},
\end{equation}
where the integrand is the so called mass aspect function. Note that this is a slight abuse
of terminology since the quantity $m$ is physically the total energy, that is the first component of the energy-momentum 4-vector \cite{CederbaumCortierSakovich}. Definitions of the total energy, linear momentum, as well as other Hamiltonian charges, in the case of asymptotically AdS hyperbolic initial data with more general asymptotics, may be found in \cite{CederbaumCortierSakovich,ChruscielMaertenTod,ChruscielTod} (see \cite{Michel} for invariance).

Initial data modeling an asymptotically umbilical, spacelike hypersurface, in an asymptotically flat spacetime will be referred to as asymptotically hyperboloidal. These type of data are defined by the asymptotics \eqref{2} for the metric $g$, and the asymptotics \eqref{3} for the difference of the second fundamental form and the hyperbolic metric $k-g_{0}$. In this case the tensor $\mathbf{m}^{k}$ plays a role in the definition of mass
\begin{equation}\label{3.0'}
m=\frac{1}{16\pi}\int_{S^{2}}\left[Tr_{\sigma}\left( \mathbf{m}^{g} + 2\mathbf{m}^{k} \right)+2\mathbf{m}^{r}\right],
\end{equation}
where again the integrand is referred to as the mass aspect function.

The initial data are required to satisfy the constraint equations
\begin{equation}\label{4}
	2\mu = R+(Tr_{g} k)^{2}-|k|_{g}^{2}-2\Lambda,
	\text{ }\text{ }\text{ }\text{ }\text{ }\text{ } J = \operatorname{div}_{g}(k- (Tr_{g} k) g),
\end{equation}
where $\mu$ and $J$ are the energy and momentum density of the matter fields, $\Lambda$ is a cosmological constant and $R$ is scalar curvature. In the asymptotically AdS hyperbolic setting $\Lambda=-3$, whereas in the asymptotically hyperboloidal context $\Lambda=0$.
Furthermore, in both cases the dominant energy condition is expressed by
\begin{equation}\label{5}
\mu \geq |J|_{g}.
\end{equation}

Electromagnetic charge may be included into certain geometric inequalities by considering initial data for the Einstein-Maxwell equations $(M,g,k,E,B)$. The electric and magnetic fields induced on the slice are denoted by $E$ and $B$, respectively. These initial data sets will also be labeled as asymptotically AdS hyperbolic or asymptotically hyperboloidal, if the electromagnetic field satisfies the following asymptotics
\begin{equation}\label{6}
  E_{r},B_{r} = O(r^{-3}),\text{ }\text{ }\text{ }\text{ }\text{ }E_{\alpha},B_{\alpha}=O(r^{-1})\text{ }\text{ }\text{ }\text{ }\text{ }
  \Rightarrow\text{ }\text{ }\text{ }\text{ }|E|_{g}+|B|_{g}=O(r^{-2}).
\end{equation}
Let $(E\times B)_{i}=\epsilon_{ijl}E^{j}B^{l}$ denote the cross product with
$\epsilon$ the volume form of $g$, then the energy and momentum density of the non-electromagnetic matter fields is given by
\begin{equation}\label{7}
  	2\mu_{EM} = R+(Tr_{g} k)^{2}-|k|_{g}^{2} - 2(|E|_{g}^2+|B|_{g}^{2})-2\Lambda,\text{ }\text{ }\text{ }\text{ }\text{ }\text{ } J_{EM} = \operatorname{div}_{g}(k- (Tr_{g} k) g)+2E\times B,
\end{equation}
with the divergence of $E$ and $B$ interpreted as the electric and magnetic charge density. The following inequality will be referred to as the charged dominant energy condition
\begin{equation}\label{10}
   \mu_{EM} \geq |J_{EM}|_{g} + \frac{1}{2}\left(|\operatorname{div}_{g}E|
   +|\operatorname{div}_{g}B|\right).
\end{equation}
Observe that
\begin{equation}\label{8}
\mathcal{Q}_{e} = \frac1{4\pi} \int_{S_\infty} g(E,\nu_{g})\, , \qquad
\mathcal{Q}_{b} = \frac1{4\pi} \int_{S_\infty} g(B,\nu_{g})\, ,
\end{equation}
are well-defined under the conditions \eqref{6};
here $S_{\infty}$ represents the limit as $r\rightarrow\infty$ of integrals over coordinate spheres $S_{r}$, with outer unit normal $\nu_{g}$. The quantities $\mathcal{Q}_{e}$ and $\mathcal{Q}_{b}$ are the total electric and magnetic charge respectively,
and the square of the total charge is $\mathcal{Q}^{2}=\mathcal{Q}_{e}^{2}+\mathcal{Q}_{b}^{2}$.

The initial data may have a boundary $\partial M$, which will typically consist of an outermost apparent horizon. This means that each component $S\subset\partial M$ satisfies $\theta_{+}(S):=H_{S}+Tr_{S}k=0$ (future horizon) or $\theta_{-}(S):=H_{S}-Tr_{S}k=0$ (past horizon), and that no other apparent horizons are present; $H$ is the mean curvature with respect to the normal pointing towards the designated asymptotically hyperbolic end. Moreover, in some situations the initial data will have other ends which are not asymptotically hyperbolic, but are rather asymptotically flat or asymptotically cylindrical. Recall that an asymptotically flat end is diffeomorphic to $\mathbb{R}^{3}\setminus\mathrm{Ball}$, and in the Cartesian coordinates $x^i$ given by this diffeomorphism the following decay conditions are present
\begin{equation}\label{9}
|g_{ij}-\delta_{ij}|+|x||\partial g_{ij}| + |x|^2|\partial\partial g_{ij}|=O(|x|^{-1}),\text{ }\text{ }\text{
}\text{ }
|k_{ij}|+ |E_{i}|+|B_{i}|=O(|x|^{-2})\text{ }\text{ }\text{as}\text{ }\text{
}|x|\rightarrow\infty.
\end{equation}
%\begin{equation}\label{9}
%|(g_{ij}-\delta_{ij})|+|x||\partial g_{ij}|=O(|x|^{-1}),\text{ }\text{ }\text{
%}\text{ }|k_{ij}|=O(|x|^{-2}),\text{ }\text{ }\text{
%}\text{ }|E_{i}|+|B_{i}|=O(|x|^{-2})\text{ }\text{ }\text{as}\text{ }\text{
%}|x|\rightarrow\infty.
%\end{equation}
For such an end, the ADM mass is well-defined and given by
\begin{equation}\label{10.1}
m_{adm}=\frac{1}{16\pi}\int_{S_{\infty}}(\partial_{i}g_{ij}-\partial_{j}g_{ii})\nu^{j}.
\end{equation}
In the case of cylindrical ends, the asymptotics are most easily described in terms of Brill coordinates, see Section \ref{sec7}. Unless it is stated otherwise, initial data sets will be assumed to have only one end.

In each of the following sections, a reduction procedure will be introduced for a geometric inequality associated with asymptotically AdS hyperbolic initial data. These inequalities include the positive mass theorem (with and without charge), the Penrose inequality (with and without charge), and the mass-angular momentum inequality (with and without charge). Thus, each inequality is reduced to solving a canonical system of PDE. Furthermore, we show that the primary (or Jang-type) equations for each system may be solved independently with the desired asymptotics.\medskip

\textbf{Acknowledgements.} The first author would like to thank Anna Sakovich for helpful discussions.

\section{The Positive Mass Theorem}
\label{sec3} \setcounter{equation}{0}
\setcounter{section}{3}

Let $(M,g,k)$ be a smooth asymptotically AdS hyperbolic initial data set satisfying the dominant energy condition \eqref{5}, which either has an apparent horizon boundary or no boundary. Then the positive mass theorem states that
\begin{equation} \label{3.1}
m \geq 0,
\end{equation}
and equality holds if and only if $(M,g,k)$ arises from an embedding into anti-de Sitter space. Initial investigations were carried out in \cite{AbbottDeser,AshtekarMagnon,GibbonsHullWarner,HenneauxTeitelboim}, the time symmetric case ($k=0$) was treated in \cite{AnderssonCaiGalloway,ChruscielHerzlich,ChruscielJezierskiLeski,Wang,Zhang}, and the most general versions were given in \cite{ChruscielMaertenTod,Maerten,WangXieZhang,XieZhang}. See also \cite{ChenHungWangYau} for a discussion concerning the rigidity statement. Except for \cite{AnderssonCaiGalloway}, the methods used to prove \eqref{3.1} have relied on the spinor approach generalizing that of Witten \cite{Witten} in the asymptotically flat setting. In this section we propose a different strategy based on deformations of the initial data which result in an asymptotically flat structure. The deformations also preserve the mass and yield `sufficiently' nonnegative scalar curvature, so that \eqref{3.1} follows from the proofs of positivity of the ADM mass given by Schoen and Yau \cite{SchoenYau0,SchoenYau}. Thus the positive mass theorem in the asymptotically AdS hyperbolic setting is reduced to the positive mass theorem in the ADM context, whenever a canonical system of elliptic equations admits a solution. Moreover we show that when decoupled, each equation in the system possesses solutions with the desired asymptotics.

The deformation procedures consist of three parts. First, similar to \cite{BrayKhuri1,BrayKhuri2}, the given asymptotically AdS hyperbolic data $(M,g,k)$ is transformed into a time symmetric data set $(M_1,g_1)$ which is also asymptotically AdS hyperbolic, has the same mass, and satisfies the scalar curvature lower bound $R_1\geq -6$ weakly. Second, following \cite{ChruscielTod}, $(M_1,g_1)$ is transformed into an umbilic asymptotically hyperboloidal data set $(M_2,g_2,k_2)$ which again preserves the mass and satisfies $R_2\geq-6$ weakly. Lastly, in the spirit of \cite{ChaKhuriSakovich,Sakovich,SchoenYau1}, the hyperboloidal data $(M_2,g_2,k_3)$ is deformed into a time symmetric asymptotically flat data set $(M_3,g_3)$, having twice the original mass and nonnegative scalar curvature $R_3\geq 0$ weakly.

We now describe the first deformation. Let $M_{1}=\{t=f(x)\}$ be the graph embedded in the warped product 4-manifold $(M \times \mathbb{R}, g + u^2 dt^2)$, which satisfies the generalized Jang equation
\begin{equation}\label{3.2}
\left(g^{ij}-\frac{u^{2}f^{i}f^{j}}{1+u^{2}|\nabla f|_{g}^{2}}\right)
\left(\frac{u\nabla_{ij}f+u_{i}f_{j}+u_{j}f_{i}}
{\sqrt{1+u^{2}|\nabla f|_{g}^{2}}}-k_{ij}\right)=0.
\end{equation}
Then set $g_{1}=g+u^{2}df^{2}$ to be the induced metric on this graph.
Here $u$ is a positive function that will be chosen appropriately.
The purpose of the generalized Jang equation is to impart a desired lower bound for the
scalar curvature of $g_{1}$. Namely, as is shown in \cite{BrayKhuri1,BrayKhuri2} we have
\begin{equation}\label{3.3}
2 \mu_{1} := R_{1} - 2 \Lambda=2(\mu-J(w))+
|\pi-k|_{g_{1}}^{2}+2|q|_{g_{1}}^{2}
-2u^{-1}\operatorname{div}_{g_{1}}(u q),
\end{equation}
where $\mu_{1}$ is the energy density of the new data $(M_1,g_1)$,
$\pi$ denotes the extrinsic curvature of $M$ in the dual Lorentzian setting $(M_{1} \times \mathbb{R}, g_{1}-u^2 dt^2)$,
and $w$ and $q$ are 1-forms given by
\begin{equation}\label{3.4}
\pi_{ij}=\frac{ u \nabla_{ij}f
+ u_i f_j +  u_j  f_i}{ \sqrt{1 + u^2 |\nabla f|_g^2 }},\text{ }\text{ }\text{ }\text{ }
w_{i}=\frac{u f_{i}}{\sqrt{1+u^{2}|\nabla f|_{g}^{2}}},\text{
}\text{ }\text{ }\text{ }
q_{i}=\frac{u f^{j}}{\sqrt{1+u^{2}|\nabla f|_{g}^{2}}}(\pi_{ij}-k_{ij}).
\end{equation}
Observe that the dominant energy condition \eqref{5} together with \eqref{3.3} shows that $R_1\geq-6$ modulo a divergence term. Furthermore in order to preserve the asymptotic geometry, and motivated by the model examples of asymptotically totally geodesic slices in AdS space, the following expansions will be imposed

\begin{equation}\label{3.5}
u= \sqrt{1+r^{2}} +  u_{0} + O_{3}\left(r^{-1 + \varepsilon}\right),\quad\quad\quad
f = O_{4}\left(r^{-3}\right),
\end{equation}
where $u_{0}$ is a (mass related) constant, and $\varepsilon>0$.

\begin{lemma}\label{lemma3.1}
If $(M,g,k)$ is asymptotically AdS hyperbolic and \eqref{3.5} is satisfied, then $(M_1,g_1)$ is asymptotically AdS hyperbolic, and the mass is given by $m_{1}=m$.
\end{lemma}

\begin{proof}
It is clear that $(M_{1},g_{1})$ has an end diffeomorphic to $S^2 \setminus [r_0,\infty)$. Moreover \eqref{1}, \eqref{2}, and \eqref{3.5} yield
\begin{equation}\label{3.7}
\begin{split}
  & (g_{1})_{rr} = g_{rr}+u^{2}f_{r}^{2}= \frac{1}{1+r^{2}} + \frac{\mathbf{m}^{r} }{r^{5}} + O_{3}(r^{-6}),  \text{ }\text{ }\text{ } \\
  & (g_{1})_{r\alpha} =g_{r\alpha}+u^{2}f_{r}f_{\alpha}= O_{3}(r^{-3}),
\text{ }\text{ }\text{ } \\
  & (g_{1})_{\alpha \beta} = g_{\alpha\beta}+u^{2}f_{\alpha}f_{\beta}
= r^{2} \sigma_{\alpha \beta}+ \frac{\mathbf{m}^{g}_{\alpha\beta} }{r}+ O_{3}(r^{-2}).
\end{split}
\end{equation}
It follows that the mass aspect function of the new data unchanged from the original, and therefore $m_{1}=m$.
\end{proof}

In \cite{ChruscielTod}, a transformation is exhibited which takes a constant mean curvature asymptotically hyperboloidal data set into a maximal asymptotically AdS hyperbolic data set.
Here we perform this transformation in the opposite direction to obtain an umbilic asymptotically hyperboloidal data set from a time symmetric asymptotically AdS hyperbolic data set. Namely, consider the new data $(M_{2},g_{2},k_{2})$ defined by
\begin{equation} \label{3.8}
(M_{2}, g_{2}) \equiv (M_{1}, g_{1}), \quad\quad\text{ }
k_{2} = \sqrt{\frac{-\Lambda}{3}} g_{1},
\end{equation}
with corresponding matter energy and momentum density
\begin{equation} \label{3.11}
\begin{split}
& 2\mu_{2} = R_{2} + (Tr_{g_{2}} k_{2})^{2} - |k_{2}|_{g_{2}}^{2}
= R_{1} -2 \Lambda = 2 \mu_{1} ,  \\
& J_{2} = \operatorname{div}_{g_{2}}(k_{2}- (Tr_{g_{2}} k_{2}) g_{2}) =0.
\end{split}
\end{equation}

\begin{lemma}\label{lemma3.2}
If $(M_{1}, g_{1})$ is asymptotically AdS hyperbolic then $(M_{2}, g_{2}, k_{2})$ is asymptotically hyperboloidal with mass given by $m_{2}=m_{1}$.
\end{lemma}

\begin{proof}
We have $k_{2} = g_{1} = g_{2}$, since $\Lambda = -3$. It is then clear that $(M_{2}, g_{2}, k_{2})$ is asymptotically hyperboloidal. Furthermore $\mathbf{m}_{2}^{k}=\mathbf{m}_{1}^{g}=\mathbf{m}_{2}^{g}$, and hence
\begin{equation} \label{3.10}
m_{2}
= \frac{1}{16\pi} \int_{S^{2}} \left[ Tr_{\sigma}\left( \mathbf{m}^{g}_{2} + 2\mathbf{m}^{k}_{2}\right) + 2\mathbf{m}^{r}_{2} \right]
= \frac{1}{16\pi} \int_{S^{2}} 3 Tr_{\sigma} \mathbf{m}_{1}^{g} + 2\mathbf{m}_{1}^{r} =m_{1}.
\end{equation}
\end{proof}

The third deformation is based on \cite{Sakovich,SchoenYau1}, and yields an asymptotically
flat data set from the asymptotically hyperboloidal $(M_2,g_2,k_2)$. Consider a graph $M_{3}=\{t=\widetilde{f}(x)\}$ embedded in the product 4-manifold $(M_{2} \times \mathbb{R}, g_{2} + dt^2)$, so that the induced metric on $M_{3}$ is given by $g_{3}=g_{2}+ d\widetilde{f}^{2}$. Motivated by the model hyperboloidal slices in Minkowski space, the following asymptotics will
be imposed on the graph
\begin{equation}\label{3.12}
\widetilde{f}(r,\theta,\phi)=\sqrt{1+r^{2}}+\mathcal{A}\log r+\mathcal{B}(\theta,\phi)+\hat{f}(r,\theta,\phi),
\end{equation}
where $(\theta,\phi)$ are coordinates on $S^{2}$,
\begin{equation}\label{3.13}
\mathcal{A}=2m_{2},\text{ }\text{ }\text{ }\text{ }\text{ }\text{ }
\Delta_{\sigma}\mathcal{B}=\frac{1}{2}\left[3Tr_{\sigma} \mathbf{m}^g  +2\mathbf{m}^r\right]-\frac{1}{8\pi}\int_{S^{2}}\left[3Tr_{\sigma} \mathbf{m}^g +2\mathbf{m}^r\right],
\end{equation}
and for some $\varepsilon>0$
\begin{equation}\label{3.14}
\hat{f} = O_{4}(r^{-1+\varepsilon}).
\end{equation}
If the classical Jang equation is satisfied
\begin{equation} \label{3.15}
\left(g_{2}^{ij}-\frac{\widetilde{f}^{i}\widetilde{f}^{j}}{1+|\widetilde{\nabla} \widetilde{f}|_{g_{2}}^{2}}\right)
\left(\frac{\widetilde{\nabla}_{ij}\widetilde{f}}
{\sqrt{1+|\widetilde{\nabla} \widetilde{f}|_{g_{2}}^{2}}}- (k_{2})_{ij}\right)=0,
\end{equation}
then the scalar curvature of the Jang graph $(M_{3}, g_{3})$ enjoys weak nonnegativity through the formula \cite{SchoenYau}
\begin{equation}\label{3.16}
R_{3}=2(\mu_{2}-J_{2}(\widetilde{w}))+
|\pi_{2}-k_{2}|_{g_{3}}^{2}+2|\widetilde{q}|_{g_{3}}^{2}
- 2\operatorname{div}_{g_{3}}(\widetilde{q}),
\end{equation}
where
$\pi_{2}$ is the second fundamental form of the graph in the dual Lorentzian setting, and $\widetilde{w}$ and $\widetilde{q}$ are 1-forms given by
\begin{equation}\label{3.17}
(\pi_{2})_{ij}=\frac{ \widetilde{\nabla}_{ij} \widetilde{f}}{ \sqrt{1 + |\widetilde{\nabla} \widetilde{f}|_{g_{2}}^2 }},\text{ }\text{ }\text{ }\text{ }
\widetilde{w}_{i}=\frac{\widetilde{f}_{i}}{\sqrt{1+|\widetilde{\nabla} \widetilde{f}|_{g_{2}}^{2}}},\text{
}\text{ }\text{ }\text{ }
\widetilde{q}_{i}=\frac{ \widetilde{f}^{j}}{\sqrt{1+|\widetilde{\nabla} \widetilde{f}|_{g_{2}}^{2}}} \left((\pi_{2})_{ij}- (k_{2})_{ij}\right).
\end{equation}
Here $\widetilde{\nabla}$ denotes covariant differentiation with respect to $g_{2}$. In \cite{ChaKhuriSakovich,Sakovich} the following fact is established.

\begin{lemma} \label{lemma3.3}
If $(M_{2},g_{2},k_{2})$ is asymptotically hyperboloidal and \eqref{3.12}-\eqref{3.14} are satisfied then $(M_3,g_3)$ is asymptotically flat, and the ADM mass is given by $m_{3}=2m_{2}$.
\end{lemma}

At this stage we have an asymptotically flat initial data set $(M_3,g_3)$ which encodes the
mass of the original data. If the scalar curvature $R_3$ is nonnegative, then the
(ADM) positive mass theorem \cite{SchoenYau0} would apply to yield \eqref{3.1}. However this is not necessarily the case, since the scalar curvature is only guaranteed to be weakly nonnegative. In particular, combining
\eqref{3.3}, \eqref{3.11}, and \eqref{3.16} together leads to
\begin{equation} \label{3.18}
\begin{split}
R_{3}&=2(\mu_{2}-J_{2}(\widetilde{w}))+
|\pi_{2}-k_{2}|_{g_{3}}^{2}+2|\widetilde{q}|_{g_{3}}^{2}
-2\operatorname{div}_{g_{3}}(\widetilde{q}) \\
%&= 2(\mu-J(w))+
%|\pi-k|_{g_{1}}^{2}
%+|\pi_{2}-k_{2}|_{g_{3}}^{2}
%+2|q|_{g_{1}}^{2}
%+2|\widetilde{q}|_{g_{3}}^{2}
%-2u^{-1}\operatorname{div}_{g_{1}}(u q)
%-2\operatorname{div}_{g_{3}}(\widetilde{q}) \\
&= 2(\mu-J(w))+
|\pi-k|_{g_{2}}^{2}
+|\pi_{2}-k_{2}|_{g_{3}}^{2}
+2|q|_{g_{2}}^{2}
+2|\widetilde{q}|_{g_{3}}^{2}
-2u^{-1}\operatorname{div}_{g_{2}}(u q)
-2\operatorname{div}_{g_{3}}(\widetilde{q}).
\end{split}
\end{equation}
Nevertheless, the scalar curvature is `sufficiently' nonnegative to allow the basic strategy
in \cite{SchoenYau} to be carried out. Namely if $\psi>0$ solves
\begin{equation} \label{3.20}
\Delta_{g_{3}} \psi - \frac{1}{8} R_{3} \psi =0,
\end{equation}
with an asymptotic expansion at spatial infinity of the form
\begin{equation} \label{3.21}
\psi = 1 + \frac{\psi_{0}}{r} + O_{3}(r^{-2}),
\end{equation}
where $\psi_{0}$ is a constant, then $(M_4,g_4)=(M_3,\psi^4 g_3)$ is asymptotically flat with zero scalar curvature. By \cite{SchoenYau0} the mass of the conformal metric is nonnegative $m_4\geq 0$. Therefore, since $m_4=m_3+2\psi_0$ it remains to show that $\psi_0\leq -\frac{1}{4}m_3$. In order to accomplish this, and to aid with existence and positivity of $\psi$ we choose
\begin{equation} \label{3.19}
u = \psi^{2} \sqrt{1 + |\widetilde{\nabla} \widetilde{f}|_{g_{2}}^{2}}.
\end{equation}
In particular, $u$ follows the desired asymptotics in \eqref{3.5}, with $u_{0}=2 \psi_{0}+ \mathcal{A}$.

\begin{theorem}\label{thm3.4}
Let $(M, g, k)$ be a $3$-dimensional, complete, asymptotically AdS hyperbolic initial data set satisfying the dominant energy condition $\mu\geq|J|$.
If the coupled Jang system of equations \eqref{3.2}, \eqref{3.15}, \eqref{3.20}, and \eqref{3.19} admits a global smooth solution satisfying the asymptotics \eqref{3.5}, \eqref{3.12}-\eqref{3.14}, and \eqref{3.21}, then $m\geq 0$. Moreover $m=0$ if and only if the initial data arise from an embedding into the anti-de Sitter spacetime.
\end{theorem}

\begin{proof}
As discussed above, in order to establish the inequality $m\geq 0$ it remains to show that $\psi_0\leq -\frac{1}{4}m_3$. For this purpose multiply \eqref{3.20} by $\psi$ and integrate by parts, while utilizing \eqref{3.18}, to yield
\begin{equation} \label{3.22}
\begin{split}
\int_{\widetilde{S}_{\infty}} 4g_{3}(\psi \nabla \psi, \nu_{3})
&= \int_{M_{3}} 4|\nabla_{g_{3}} \psi|^{2} + \frac{1}{2}R_{3} \psi^{2} \\
&\geq \int_{M_{3}} 4|\nabla_{g_{3}} \psi|^{2}
+ \psi^{2} \left( |\widetilde{q}|_{g_{3}}^{2}
-\operatorname{div}_{g_{3}}(\widetilde{q}) \right)
+ \psi^{2} \left(|q|_{g_{2}}^{2}- u^{-1}\operatorname{div}_{g_{2}}(u q)\right)
 \\
&\geq  \int_{M_{3}} 3|\nabla_{g_{3}} \psi|^{2}
- \int_{\widetilde{S}_{\infty}} \psi^{2} g_{3}(\widetilde{q}, \nu_{3})
 + \int_{M_{3}} \psi^{2} \left(|q|_{g_{2}}^{2}- u^{-1}\operatorname{div}_{g_{2}}(u q)\right),
\end{split}
\end{equation}
where $\nu_{3}$ is the unit outer normal with respect to $g_{3}$ and $\widetilde{S}_r$ denotes
coordinate spheres in $M_3$. By Lemma 10.1 in \cite{ChaKhuriSakovich} and Lemma \ref{lemma3.3}
\begin{equation} \label{3.23}
\int_{\widetilde{S}_{\infty}} \psi^{2} g_{3}(\widetilde{q}, \nu_{3}) =-4\pi\mathcal{A}
=-8\pi m_2=-4\pi m_3.
\end{equation}
Now recall the relation between volume forms that
\begin{equation} \label{3.24}
d\omega_{g_{3}}  = \sqrt{1 + |\widetilde{\nabla} \widetilde{f}|_{g_{2}}^{2}} d\omega_{g_{2}}.
\end{equation}
This, together with the choice of $u$ in \eqref{3.19} produces
\begin{equation} \label{3.25}
\int_{M_{3}} \psi^{2} u^{-1}\operatorname{div}_{g_{2}}(u q) d\omega_{g_{3}}
= \int_{M_{2}}  \operatorname{div}_{g_{2}}(u q) d\omega_{g_{2}} =
\int_{\overline{S}_{\infty}}ug_{2}(q,\nu_2)=0,
\end{equation}
where the last equality is shown in Lemma \ref{lemma9.1} below and $\overline{S}_r$ denotes
coordinate spheres in $M_2$.
It follows that
\begin{equation} \label{3.26}
- 16\pi \psi_{0}
= \int_{ \widetilde{S}_{\infty}} 4g_{3}(\psi \nabla_{g_{3}} \psi, \nu_{3})
\geq 4\pi m_3,
\end{equation}
which is the desired result.

Now consider the case of equality. If $m=0$ then $\psi_0=m_{1}=m_{2}=m_{3}=0$, and in particular \cite{SchoenYau0} shows that $(M_{3}, g_{3})$ is isometric to $(\mathbb{R}^{3}, \delta)$. It follows that $(M_{2}, g_{2}, k_{2})$ is a graphical totally umbilical slice of Minkowski space, where the graph has an asymptotic expansion of the form $t=\widetilde{f}= \sqrt{1+r^{2}}+B(\theta, \phi) + \hat{f}$. Observe that since $g_2$ is asymptotically hyperboloidal $(g_{2})_{\alpha \beta} = r^{2}\sigma_{\alpha\beta} + O(r^{-1})$; on the other hand
\begin{equation} \label{3.27}
(g_{2})_{\alpha \beta}=(g_3)_{\alpha\beta} - \widetilde{f}_{\alpha} \widetilde{f}_{\beta}
=  \delta_{\alpha \beta} - \widetilde{f}_{\alpha} \widetilde{f}_{\beta}  \\
= r^{2}\sigma_{\alpha\beta} - \mathcal{B}_{\alpha}\mathcal{B}_{\beta} + O(r^{-1+\varepsilon}).
\end{equation}
Therefore $|\nabla\mathcal{B}|_{\sigma} = O(r^{-\frac{1}{2}+ \frac{\varepsilon}{2}})$, and since $\mathcal{B}$ is independent of $r$ this yields $\mathcal{B} = const$.
Furthermore, a similar argument applies in the radial direction to produce
\begin{equation} \label{3.28}
 \frac{1}{1+r^{2}} + O(r^{-5})=(g_{2})_{r r} =
 \delta_{r r} - (\widetilde{f}_{r})^{2}
= \frac{1}{1+r^{2}} -2 \hat{f}_{r} + O(r^{-4+2\varepsilon}).
\end{equation}
This together with the fact that $\hat{f}=O(r^{-1+\varepsilon})$, implies $\hat{f} = O(r^{-3+2\epsilon})$. In \cite{Choquet-Bruhat2}, it is shown that if a constant mean curvature spacelike graph $t=h(r, \theta, \phi)$ embedded in Minkowski space is sufficiently close to the hyperboloid at infinity ($h=\sqrt{1+r^{2}} + o(r^{-1})$), then it must be isometric to the hyperboloid. Hence, $\hat{f}\equiv 0$ and $(M_{2}, g_{2})\equiv (M_{1}, g_{1})$ is isometric to hyperbolic 3-space.

The original data $(M,g,k)$ must then arise from an embedding into the AdS spacetime via the graph $t=f(x)$.  To see this, observe that $g=g_1-u^2 df^2=g_0-u^2 df^2$. Moreover
$\psi_0=0$ and \eqref{3.22} gives $\psi\equiv 1$, so that together with $\widetilde{f}=\sqrt{1+r^2}+\mathcal{B}$ we have $u = \sqrt{1+r^{2}}$. Thus, $(M,g)$ is isometrically embedded into AdS space. Moreover $k$ agrees with the second fundamental form of the embedding, since \eqref{3.18} and \eqref{3.22} give $|\pi-k|_{g_2}=0$.
\end{proof}

It should be pointed out that the hypotheses in Theorem \ref{thm3.4} should be generalized to allow for initial data with apparent horizon boundary. This could be accomplished by allowing for blow-up of the Jang type equations \eqref{3.2}, \eqref{3.15} at apparent horizons as in \cite{HanKhuri1,SchoenYau}. We expect that the cylindrical ends which arise from such blow-up behavior can be conformally closed, so that the remaining arguments of the proof remain valid as in \cite{SchoenYau}.

We also remark that strong evidence suggests that a positive solution to \eqref{3.20}, \eqref{3.21} exists,
as a result of the choice of $u$ in \eqref{3.19}. To see this, observe that the same computation as in \eqref{3.22} shows that the kernel for this problem is trivial. Moreover
arguing by contradiction, suppose that the domain on which $\psi<0$ is nonempty and observe that it must be a bounded set in light of the asymptotics \eqref{3.21}. Then integrating by parts on this domain as in \eqref{3.22} yields a contradiction. The Hopf lemma may then be used to rule out any points at which $\psi$ vanishes.

In order to lend further credence to the above procedure we show here that solutions to the Jang equations \eqref{3.2} and \eqref{3.15} exist with the desired asymptotics. First note that \eqref{3.15} has been thoroughly analyzed in \cite{SakovichThesis,Sakovich}, and in particular solutions are shown to exist satisfying \eqref{3.12}-\eqref{3.14}. Equation \eqref{3.2} will be treated in the proposition below.
For this we assume that $(g,k)$ has the form \eqref{1} with $a$ as in \eqref{2} satisfying the extra conditions $a_{rr} = a_{r\alpha} = 0$; in this case $\mathbf{m}^{r}=0$. We point out that if $(g,k)$ is sufficiently regular with $|g - g_0|_{g_{0}} = O(r^{-3})$ and $|k|_{g_{0}} = O(r^{-3})$, one can perform a change of coordinates at infinity as described in Appendix A of \cite{ChaKhuriSakovich} to achieve $a_{rr} = a_{r\alpha} = 0$, and this change of coordinates does not affect the mass aspect function.

\begin{proposition} \label{thm3.5}
Suppose that $(M,g,k)$ has an outermost apparent horizon boundary, and that $u$ is a
given smooth positive function $u$ satisfying \eqref{3.5}. Then there exists a smooth solution to the generalized Jang equation \eqref{3.2} with the fall-off \eqref{3.5}. Moreover, the solution
blows-up (down) at the future (past) apparent horizon boundary components, and
precise asymptotics at the horizon are given in \cite{HanKhuri1}.
\end{proposition}

\begin{proof}
The existence of a solution may be established by suitable modifications of the arguments
in \cite{SakovichThesis,Sakovich}. What remains is to find appropriate barriers in order to
show the desired asymptotics. At the horizons, one may take the barriers given in \cite{HanKhuri1}. Below, barriers are constructed at infinity following \cite{SakovichThesis,Sakovich}.

We search for a radial function
\begin{equation} \label{3.51}
p(r)=\frac{\overline{f}'(r)(1+r^2)}{\sqrt{1+(1+r^2)^{2}\overline{f}'(r)^2}},
\end{equation}
from which a barrier function $\overline{f}(r)$ will be determined. Note that $-1\leq p \leq 1$, and $p= \pm 1$ when $\overline{f}'=\pm \infty$. Let $\mathcal{F}(\overline{f})$ denote the left-hand side of the generalized Jang equation \eqref{3.2} evaluated at $f = \overline{f}$, and set
\begin{equation}\label{3.52}
\begin{split}
\beta(r, \theta, \phi) =&\frac{\frac{1}{1+r^{2}}+(1+r^2)\overline{f}'(r)^2}{u^{-2}+(1+r^2)\overline{f}'(r)^2}\\
                        =& \frac{1}{1 + (1-p^2)\left((1+r^{2})u^{-2}-1 \right)}\\
                        =& 1 + \frac{2 u_{0}}{r} (1-p^2) + O(r^{-2+\varepsilon}).
\end{split}
\end{equation}
Then a computation shows that
\begin{equation} \label{3.53}
\begin{split}
&\frac{\mathcal{F}(\overline{f})}{\beta^{3/2}(1 + r^2)^{3/2}u^{-2}} \\
=&  p'  + \left(\frac{3}{r}-\frac{2u_{0}}{r^{2}}+O(r^{-3+\varepsilon})\right)p +O(r^{-4})p^{2}
+ \left(-\frac{1}{r} + \frac{6u_{0}}{r^{2}} + O(r^{-3}) \right) p^{3} +O(r^{-4}),
\end{split}
\end{equation}
Now choose $p_+$ and $p_-$ to be solutions of the boundary value problems
\begin{equation} \label{3.54}
p_{\pm}'
+ \left( \frac{3}{r}-\frac{2u_{0}}{r^{2}} \right)p_{\pm} \pm \frac{C_{1}}{r^{3-\varepsilon}}|p_{\pm}|
\pm \frac{C_{2}}{r^{4}} p_{\pm}^{2}
+ \left(-\frac{1}{r} + \frac{6u_{0}}{r^{2}} \right) p_{\pm}^{3} \pm \frac{C_{3}}{r^{3}} |p_{\pm}|^{3}
\pm \frac{C_{4}}{r^{4}} \\
=0,
\end{equation}
\begin{equation} \label{3.55}
p_{\pm}(r_0)  = \mp 1,
\end{equation}
where $C_i$, $i=1,\ldots,4$, are positive constants. A similar analysis as conducted in \cite{SakovichThesis, Sakovich} applies to this system to yield $p_{\pm}=O(r^{-2})$. This
implies that the resulting barriers satisfy $\overline{f}_{\pm}=O(r^{-3})$. The remainder of the proof may be carried out as in \cite{SakovichThesis, Sakovich}.

%The only difficulty is
%making sure that the barriers constructed still work with the lower %order terms $\phi_{i}f_{j}$. However
%since $\phi_{i}=O(r^{-2})$, this should not be a problem. [[Of course %more details are needed here.]]

\end{proof}

\section{The Penrose Inequality}
\label{sec4} \setcounter{equation}{0}
\setcounter{section}{4}

The heuristic arguments of Penrose \cite{Penrose, Penrose1}, if assumed to be valid in the asymptotically AdS hyperbolic setting, lead to the following version of the Penrose inequality
\begin{equation}\label{4.1}
m\geq\sqrt{\frac{A}{16\pi}} + \frac{1}{2}\left( \sqrt{\frac{A}{4\pi}} \right)^{3},
\end{equation}
where $A$ is the minimum area required to enclose the outermost apparent horizon. The rigidity statement asserts that equality is achieved precisely for initial data which arise from an embedding into the Schwarzschild-AdS spacetime. Unfortunately the heuristic arguments, which depend crucially on black hole stability, are most likely not applicable in the AdS case since
the AdS spacetime itself is believed to be unstable \cite{Bizon,BizonRostworowski,DafermosHolzegel}. Thus, the inequality \eqref{4.1} which is not even known to be true for time symmetric data, is not particularly well-motivated.

Here we propose an alternate version of the Penrose inequality in the asymptotically AdS hyperbolic realm. Namely
\begin{equation}\label{4.1a}
m \geq \sqrt{\frac{A}{16\pi}},
\end{equation}
where $A$ is the minimum area required to enclose the outermost mean convex surface $S$ satisfying
\begin{equation}\label{4.1b}
\theta_{+}(S)\theta_{-}(S)=H_{S}^2-(Tr_{S}k)^2=4
\end{equation}
in an asymptotically AdS hyperbolic initial data set $(M,g,k)$ with dominant energy condition \eqref{5}; the rigidity statement here is identical with that of \eqref{4.1}. In the time symmetric case this reduces to the conjecture in \cite{Wang}. Furthermore, the structure of \eqref{4.1a} is analogous to the traditional Penrose inequality \cite{Bray,HuiskenIlmanen} within the context of the asymptotically flat setting. In fact the motivation for \eqref{4.1a} and \eqref{4.1b} arises from the asymptotically flat regime through a deformation procedure described below.

As in the previous section, we seek three deformations of the initial data which yield weak nonnegativity of the scalar curvature and preserve the mass in an appropriate sense. In the current context certain aspects of the boundary geometry must also be controlled throughout the process, in particular the mean curvature. It will be assumed that the initial data has a single component boundary consisting of an outermost mean convex surface satisfying \eqref{4.1b}.
The first step is to solve the generalized Jang equation \eqref{3.2} with a warping factor $u$ to be suitably chosen, where the solution $f$ and warping function satisfy the asymptotics \eqref{3.5}. According to Lemma \ref{lemma3.1} this yields a new time symmetric, asymptotically AdS hyperbolic, initial data set $(M_1,g_1)$ with the same mass $m_1=m$.
Moreover, in order to make contact with apparent horizons in the next deformation, geometric boundary conditions are imposed when solving the generalized Jang equation, namely
\begin{equation}\label{4.1c}
H_{\partial M_{1}}=2,\quad\quad\quad q(\nu_{g_{1}})=0\text{ }\text{ on }\text{ }\partial M_{1},
\end{equation}
where $q$ is defined in \eqref{3.4}. The property \eqref{4.1b} of $\partial M$ is chosen to facilitate the boundary conditions \eqref{4.1c}.  To see this consider the case of spherically symmetric data, where a direct calculation shows that
\begin{equation}
H_{\partial M_{1}}=\frac{H_{\partial M}}{\sqrt{1+(u\partial_{\nu_{g}}f)^2}},\quad\quad\quad
q(\nu_{g_{1}})=-\frac{(u\partial_{\nu_{g}}f)^2 H_{\partial M}}{\sqrt{1+(u\partial_{\nu_{g}}f)^2}}
+(u\partial_{\nu_{g}}f) Tr_{\partial M}k.
\end{equation}
If \eqref{4.1b} is valid, then the Neumann-type boundary condition $u\partial_{\nu_{g}}f=\frac{1}{2}Tr_{\partial M}k$ guarantees \eqref{4.1c}.

The second step in the deformation agrees with that in the previous section, and is given by \eqref{3.8}. According to Lemma \ref{lemma3.2} we obtain an asymptotically hyperboloidal data set $(M_{2}, g_{2}, k_{2})$ with preserved mass $m_{2}=m_{1}$. In addition, the boundary condition \eqref{4.1c} gives a past apparent horizon boundary $H_{\partial M_{2}}-Tr_{\partial M_{2}}k_{2}=0$.

The third and final step in the deformation procedure is a generalization of that which was given in the previous section. It yields a time symmetric, asymtotically flat
initial data set $(M_{3}, g_{3})$ which we now describe. Consider a graph $M_{3}=\{t=\widetilde{f}(x)\}$ embedded in the warped product 4-manifold $(M_{2} \times \mathbb{R}, g_{2} + \widetilde{u}^2 dt^2)$, where $\widetilde{f}$ satisfies the asymptotics
\eqref{3.12}, \eqref{3.13} and \eqref{3.14}, and the function $\widetilde{u}$ is nonnegative
with asymptotic expansion
\begin{equation}\label{4.5}
\widetilde{u}= 1 + \frac{\widetilde{u}_{0}}{r}+O_{3}\left(\frac{1}{r^{2-\varepsilon}}\right),
\end{equation}
for constants $\varepsilon>0$ and $\widetilde{u}_{0}$.
The induced metric on $M_{3}$ is given by $g_{3}=g_{2}+\widetilde{u}^{2}d\widetilde{f}^{2}$.
If the generalized Jang equation
\begin{equation}\label{4.2}
\left(g_{2}^{ij}-\frac{\widetilde{u}^{2}\widetilde{f}^{i}\widetilde{f}^{j}}{1+\widetilde{u}^{2} |\widetilde{\nabla} \widetilde{f}|_{g_{2}}^{2}}\right)
\left(\frac{\widetilde{u}\widetilde{\nabla}_{ij}\widetilde{f}
+\widetilde{u}_{i}\widetilde{f}_{j}+\widetilde{u}_{j}\widetilde{f}_{i}}
{\sqrt{1+\widetilde{u}^{2}|\widetilde{\nabla} \widetilde{f}|_{g_{2}}^{2}}}-(k_{2})_{ij}\right)=0
\end{equation}
is satisfied, then the scalar curvature $R_{3}$ of the Jang graph $(M_{3}, g_{3})$ is weakly nonnegative as follows
\begin{equation}\label{4.3}
R_{3}=2(\mu_{2}-J_{2}(\widetilde{w}))+
|\pi_{2}-k_{2}|_{g_{3}}^{2}+2|\widetilde{q}|_{g_{3}}^{2}
-2\widetilde{u}^{-1}\operatorname{div}_{g_{3}}(\widetilde{u} \widetilde{q}),
\end{equation}
where
$\pi_{2}$ is the second fundamental form of of the graph in the dual Lorentzian setting,
and $\widetilde{w}$ and $\widetilde{q}$ are 1-forms given by
\begin{equation}\label{4.4}
(\pi_{2})_{ij}
=\frac{ \widetilde{u} \widetilde{\nabla}_{ij}\widetilde{f}
+ \widetilde{u}_i \widetilde{f}_j +  \widetilde{u}_j  \widetilde{f}_i}{ \sqrt{1 + \widetilde{u}^2 |\widetilde{\nabla} \widetilde{f}|_{g_{2}}^2 }},
\text{ }\text{ }\text{ }\text{ }
\widetilde{w}_{i}=\frac{\widetilde{u} \widetilde{f}_{i}}{\sqrt{1+\widetilde{u}^{2}|\widetilde{\nabla} \widetilde{f}|_{g_{2}}^{2}}},\text{
}\text{ }\text{ }\text{ }
\widetilde{q}_{i}=\frac{\widetilde{u} \widetilde{f}^{j}\left( (\pi_{2})_{ij}-(k_{2})_{ij}\right)}{\sqrt{1+\widetilde{u}^{2}|\widetilde{\nabla} \widetilde{f}|_{g_{2}}^{2}}}.
\end{equation}
Here the covariant derivative $\widetilde{\nabla}$ is with respect to $g_{2}$.

When solving \eqref{4.2}, the appropriate geometric boundary condition to impose requires
the Jang graph to have a minimal surface boundary,
\begin{equation}\label{4.1d}
H_{\partial M_{3}}=0.
\end{equation}
This typically involves blow-up of the Jang equation at the apparent
horizon boundary present in $M_{2}$.
The existence, regularity, and blow-up behavior for solutions of the generalized Jang equation was studied in \cite{HanKhuri1}, and discussed in \cite{BrayKhuri1, BrayKhuri2}.
Moreover, in \cite{ChaKhuriSakovich} it was shown that the new data possess the desired asymptotics and that the mass is preserved up to a contribution from the warping factor.

\begin{lemma}\label{lemma4.1}
If $(M_{2},g_{2},k_{2})$ is asymptotically hyperboloidal and \eqref{3.12}-\eqref{3.14}, \eqref{4.5} are satisfied then the Jang metric $g_{3}=g_{2}+ \widetilde{u}^{2}d\widetilde{f}^{2}$ is asymptotically flat, and the mass of the Jang metric is given by  $m_{3}= 2m_{2} + \widetilde{u}_{0}$.
\end{lemma}

Combining \eqref{3.3}, \eqref{3.11}, and \eqref{4.3} leads to
\begin{equation} \label{4.6}
\begin{split}
R_{3}&=2(\mu_{2}-J_{2}(\widetilde{w}))+
|\pi_{2}-k_{2}|_{g_{3}}^{2}+2|\widetilde{q}|_{g_{3}}^{2}
-2 \widetilde{u}^{-1}\operatorname{div}_{g_{3}}(\widetilde{u}\widetilde{q}) \\
&= 2(\mu-J(w))+
|\pi-k|_{g_{1}}^{2}
+|\pi_{2}-k_{2}|_{g_{3}}^{2}
+2|q|_{g_{1}}^{2}-2u^{-1}\operatorname{div}_{g_{1}}(u q)
+2|\widetilde{q}|_{g_{3}}^{2}
-2 \widetilde{u}^{-1}\operatorname{div}_{g_{3}}(\widetilde{u}\widetilde{q}) \\
&= 2(\mu-J(w))+
|\pi-k|_{g_{2}}^{2}
+|\pi_{2}-k_{2}|_{g_{3}}^{2}
+2|q|_{g_{2}}^{2}-2u^{-1}\operatorname{div}_{g_{2}}(u q)
+2|\widetilde{q}|_{g_{3}}^{2}
-2 \widetilde{u}^{-1}\operatorname{div}_{g_{3}}(\widetilde{u}\widetilde{q})
\end{split}
\end{equation}
where the final line arises from \eqref{3.8}. We will follow the arguments in \cite{BrayKhuri1,ChaKhuriSakovich} to motivate the choice of $u$ and $\widetilde{u}$.
Consider an inverse mean curvature flow (IMCF) $\{\widetilde{S}_{\tau}\}$ inside $(M_{3}, g_{3})$ emanating from the minimal boundary $\widetilde{S}_{0}=\partial M_{3}$. According to the arguments of \cite{BrayKhuri1, BrayKhuri2, ChaKhuriSakovich} together with \eqref{4.6},
\begin{equation}\label{4.8}
\begin{split}
M_{H}(\infty)-\sqrt{\frac{A}{16\pi}}
&\geq M_{H}(\infty)-M_{H}(0) \\
&\geq
-\frac{2}{(16\pi)^{3/2}}\int_{M_{3}}
\left[ u^{-1}\operatorname{div}_{g_{2}}(u q)
+\widetilde{u}^{-1}\operatorname{div}_{g_{3}}(\widetilde{u}\widetilde{q})\right]
\widetilde{H}_{\tau}\sqrt{|\widetilde{S}_{\tau}|} d\omega_{g_{3}}
\end{split}
\end{equation}
where $M_{H}$ denotes Hawking mass and $\widetilde{H}_{\tau}$, $|\widetilde{S}_{\tau}|$ are the mean curvature and area of $\widetilde{S}_{\tau}$, respectively. Recall that the relation \cite{BrayKhuri1, BrayKhuri2} between volume forms is given by
\begin{equation} \label{4.9}
d\omega_{g_{3}} = \sqrt{1 + \widetilde{u}^{2}|\widetilde{\nabla} \widetilde{f}|_{g_{2}}^{2}} d\omega_{g_{2}}.
\end{equation}
Thus, in order to apply the divergence theorem to \eqref{4.8} and avoid interior terms
we choose
\begin{equation}\label{4.7}
\widetilde{u}=\sqrt{\frac{|\widetilde{S}_{\tau}|}{16\pi}}\widetilde{H}_{\tau},\quad\quad\quad
u = \widetilde{u} \sqrt{1 + \widetilde{u}^{2}|\widetilde{\nabla} \widetilde{f}|_{g_{2}}^{2}}.
\end{equation}
It follows that
\begin{equation} \label{4.11}
M_{H}(\infty)-\sqrt{\frac{A}{16\pi}}
\geq-\frac{1}{8\pi} \left[
\int_{\overline{S}_{0}\cup \overline{S}_{\infty}} ug_{2}( q, \nu_{g_{2}} )
+ \int_{\widetilde{S}_{0}\cup \widetilde{S}_{\infty}}
\widetilde{u}g_{3}(\widetilde{q}, \nu_{g_{3}} ) \right],
\end{equation}
where $\nu_{g_{2}}$, $\nu_{g_{3}}$ are the unit outer normals with respect to
$g_{2}$ and $g_{3}$ respectively, and $\overline{S}_{\tau}$ denotes the natural projection of $\widetilde{S}_{\tau}$ to $M_{2}$.
Note that the choice of the warping factors $u$ and $\widetilde{u}$ in \eqref{4.7} gives rise to  a coupling of the generalized Jang equations \eqref{3.2}, \eqref{4.2} with the inverse mean curvature flow in $(M_{3}, g_{3})$. This is related to the Jang-IMCF system introduced in \cite{BrayKhuri1, BrayKhuri2}. Furthermore, it should be pointed out that \eqref{4.7} is consistent with the desired asymptotics. For instance from \eqref{3.7}, \eqref{3.8}, and \eqref{3.12}
\begin{equation}\label{4.15}
\begin{split}
|\widetilde{\nabla} \widetilde{f}|_{g_{2}}^{2}&=
\left( (1+r^{2}) -\frac{\mathbf{m}^{r} }{r} +O(r^{-2})\right)\left( \frac{r}{\sqrt{1+r^{2}}} + \frac{2m_{2}}{r}+O(r^{-2 + \varepsilon})\right)^{2} +O(r^{-2})  \\
&= r^{2} + 4m_{2}\sqrt{1+r^{2}} +O(r^{\varepsilon}).
\end{split}
\end{equation}
Moreover, it was shown in Section 3 of \cite{ChaKhuriSakovich} that
\begin{equation}\label{4.13}
\widetilde{u}_{0}= -M_{H}(\infty) = -m_{3}.
\end{equation}
Lemma \ref{lemma4.1} then implies that $m_2=m_3$. This together with \eqref{4.5} and \eqref{4.15} produces
\begin{equation} \label{4.16}
\begin{split}
u &= \widetilde{u} \sqrt{1 + \widetilde{u}^{2}|\widetilde{\nabla} \widetilde{f}|_{g_{2}}^{2}} \\
&= \left(1-\frac{m_{2}}{r} +O(r^{-2+\varepsilon}) \right)
\left[ 1 + \left(1-\frac{m_{2}}{r} +O(r^{-2+\varepsilon}) \right)^{2}
\left( r^{2} + 4m_{2}\sqrt{1+r^{2}} +O(r^{\varepsilon}) \right)  \right]^{1/2} \\
&= \sqrt{1+r^{2}}\left(1-\frac{m_{2}}{r} +O(r^{-2+\varepsilon}) \right)
\left(1+\frac{m_{2}}{r} +O(r^{-2+\varepsilon}) \right) \\
& = \sqrt{1+r^{2}} + O(r^{-1+ \varepsilon}),
\end{split}
\end{equation}
which is consistent with \eqref{3.5} and $u_{0}=0$.

\begin{theorem}\label{thm4.2}
Let $(M, g, k)$ be a $3$-dimensional, asymptotically AdS hyperbolic initial data set with a connected outermost mean convex boundary satisfying \eqref{4.1b}, and such that the dominant energy condition $\mu\geq|J|_{g}$ holds.
If the coupled system of equations \eqref{3.2}, \eqref{4.2}, and \eqref{4.7} admits a solution satisfying the asymptotics \eqref{3.5}, \eqref{3.12}-\eqref{3.14}, and \eqref{4.5}, with a weak IMCF in the sense of \cite{HuiskenIlmanen}, and such that the boundary conditions \eqref{4.1c}, \eqref{4.1d} are valid, then the Penrose inequality \eqref{4.1a} holds and if equality is achieved then the initial data arise from an embedding into the Schwarzschild-AdS spacetime.
\end{theorem}

\begin{remark}
It should be possible to treat the case of multiple component boundaries by coupling the generalized Jang equations with Bray's conformal flow \cite{Bray}. In \cite{HanKhuri2} this has been carried out for the asymptotically flat setting.
\end{remark}

\begin{proof}
%In the Section 3 of \cite{ChaKhuriSakovich} it was shown that
%\begin{equation}\label{4.13}
%\widetilde{u}_{0}= -M_{H}(\infty) = -m_{3}.
%\end{equation}
%This, together with Lemma \ref{lemma4.1}, implies that $m_2=m_3$.
Lemmas \ref{lemma3.1} and \ref{lemma3.2} together with \eqref{4.13} yield $m=m_1=m_2=m_3$. Consider now the boundary terms of \eqref{4.11}.
Since $\widetilde{S}_{0}$ is a minimal surface, $\widetilde{u}$ vanishes there, and the
boundary condition \eqref{4.1c} shows that the integrand over $\overline{S}_{0}$ vanishes as well,
since $g_{2}=g_{1}$. Thus, both interior boundary integrals vanish.
Moreover, Appendix B of \cite{ChaKhuriSakovich} has established that the second integral at spatial infinity also vanishes. Finally, in Appendix A of the current paper it is proven that the first integral at null infinity vanishes. The desired inequality \eqref{4.1a} now follows.

Now consider the case of equality for \eqref{4.1a}. The arguments of
\cite{BrayKhuri1, BrayKhuri2, ChaKhuriSakovich} show that
$\widetilde{u} = \sqrt{1 -\frac{2m}{r}}$, and $(M_{3}, g_{3})$, $(M_{2}, g_{2}, k_{2})$ are respectively isometric to the $t=0$ slice and a graphical totally umbilic slice of the Schwarzschild spacetime. Uniqueness of such umbilical slices is expected\footnote{Note, however, that a proof does not appear to be in the literature.}, and would imply that $\widetilde{f} = \widetilde{f}(r)$ arises from the ODE
\begin{equation} \label{4.17}
\widetilde{f}' = \frac{r}{\left( 1-\frac{2m}{r}\right)\sqrt{r^{2}+ 1-\frac{2m}{r}}}.
\end{equation}
We then have
\begin{equation} \label{4.18}
u = \widetilde{u}\sqrt{1+ \widetilde{u}^{2}|\widetilde{\nabla} \widetilde{f}|_{g_{2}}^{2}}
= \frac{\widetilde{u}}{\sqrt{1- \widetilde{u}^{2}|\nabla\widetilde{f}|_{g_{3}}^{2}}}
= \sqrt{r^{2} + 1-\frac{2m}{r}},
\end{equation}
which is the warping factor of the Schwarzschild-AdS spacetime. It is straightforward to see that $(M_{1}, g_{1}=g_{2}=g_{3}- \widetilde{u}^{2}d \widetilde{f}^2)$ is then isometric to the $t=0$ slice of Schwarzschild-AdS. Therefore $(M, g, k)$ indeed emerges from an embedding into the Schwarzschild-AdS spacetime given by the solution of the generalized Jang equation \eqref{3.2}, after imitating the case of equality arguments in \cite{BrayKhuri1,BrayKhuri2}.
\end{proof}

\section{The Penrose Inequality with Charge}
\label{sec5} \setcounter{equation}{0}
\setcounter{section}{5}

Let $(M,g,k,E)$ be an asymptotically AdS hyperbolic initial data set for the Einstein-Maxwell equations as described in Section \ref{sec2}, with $E$ divergence free. For simplicity in this section, it will be assumed that the magnetic field vanishes $B=0$.
The deformations of the previous section may be employed to reduce a version of the charged
Penrose inequality in the asymptotically AdS hyperbolic setting to the asymptotically flat
regime. Thus we have a three step transformation procedure
\begin{equation}
(M,g,k,E)\rightarrow(M_1,g_1,E_1)\rightarrow
(M_2,g_2,k_2,E_2)\rightarrow(M_3,g_3,E_3)
\end{equation}
in which the masses agree $m=m_1=m_2=m_3$, and it remains to describe how the electric field changes in each deformation.

Following \cite{ChaKhuriSakovich, DisconziKhuri} we define
\begin{equation}
(E_{1})_{i}= \frac{E_i + u^2 f_i f^j E_j}{\sqrt{1 + u^2 |\nabla f|^2_{g}}},\quad\quad
(E_{2})_{i}=(E_{1})_{i},\quad\quad
(E_{3})_{i}= \frac{(E_2)_i + \widetilde{u}^2 \widetilde{f}_i \widetilde{f}^j (E_2)_j}{\sqrt{1 + \widetilde{u}^2 |\widetilde{\nabla}  \widetilde{f}|^2_{g_{2}}}}.
\end{equation}
The results of \cite{DisconziKhuri} then guarantee that
\begin{equation}\label{5.1}
|E|_{g}\geq|E_1|_{g_{1}}=|E_2|_{g_2}\geq|E_{3}|_{g_{3}},
\end{equation}
and
\begin{align}
\begin{split}
\operatorname{div}_{g_{3}}E_{3}=&(1+\widetilde{u}^{2}|\widetilde{\nabla} f|_{g_{2}}^{2})^{-1/2}\operatorname{div}_{g_{2}}E_{2}\\
=&(1+\widetilde{u}^{2}|\widetilde{\nabla} f|_{g_{2}}^{2})^{-1/2}\operatorname{div}_{g_{1}}E_{1}\\
=&(1+\widetilde{u}^{2}|\widetilde{\nabla} f|_{g_{2}}^{2})^{-1/2}
(1+u^{2}|\nabla f|_{g}^{2})^{-1/2}\operatorname{div}_{g}E=0.
\end{split}
\end{align}
Moreover a direct calculation using the asymptotics \eqref{6}, \eqref{3.5}, \eqref{3.12}-\eqref{3.14}, and \eqref{4.5} shows that the total charge is conserved through the deformation process
\begin{equation}\label{q}
\mathcal{Q}^{3}_{e}=\mathcal{Q}^{2}_{e}=\mathcal{Q}^{1}_{e}=\mathcal{Q}_{e}.
\end{equation}

In order to utilize the charged dominant energy condition \eqref{10}, it is necessary to relate the
non-electromagnetic matter field energy density of the transformed data to that of the original.
With the notation of previous sections and \eqref{3.3}, observe that a compuatation yields
\begin{equation} \label{5.6}
\begin{split}
2 \mu^{1}_{EM}
&= R_{1} -2 |E_{1}|^{2}_{g_{1}} -2 \Lambda   \\
&=2(\mu_{EM}-J_{EM}(w))+|\pi-k|_{g_{1}}^{2}+2|q|_{g_{1}}^{2}
-2u^{-1}\operatorname{div}_{g_{1}}(u q)
+2 (|E|^{2}_{g}- |E_{1}|^{2}_{g_{1}}),
\end{split}
\end{equation}
and trivially $\mu_{EM}^{2}=\mu_{EM}^{1}$.
Then the scalar curvature formula \eqref{4.6} together with \eqref{5.6} produces
\begin{equation}\label{5.0}
\begin{split}
2 \mu^{3}_{EM}
&= R_{3} -2 |E_{3}|^{2}_{g_{3}}   \\
&=2(\mu_{EM}-J_{EM}(w))
+|\pi-k|_{g_{1}}^{2}
+|\pi_{2}-k_{2}|_{g_{3}}^{2}
+2|q|_{g_{1}}^{2}-2u^{-1}\operatorname{div}_{g_{2}}(u q) \\
& \quad +2|\widetilde{q}|_{g_{3}}^{2}
-2 \widetilde{u}^{-1}\operatorname{div}_{g_{3}}(\widetilde{u}\widetilde{q})
 +2 (|E|^{2}_{g}- |E_{1}|^{2}_{g_{1}})
+2(|E_{2}|^{2}_{g_{2}}- |E_{3}|^{2}_{g_{3}}),
\end{split}
\end{equation}
which is weakly nonnegative in light of the charged dominant energy condition \eqref{10} and \eqref{5.1}.
In addition $J^1_{EM}=J^{2}_{EM}=J^3_{EM}=0$.

\begin{theorem}\label{thm5.1}
Let $(M, g, k, E)$ be a $3$-dimensional, asymptotically AdS hyperbolic initial data set for the Einstein-Maxwell equations with a connected outermost mean convex boundary satisfying \eqref{4.1b}, without charged matter, and such that the charged dominant energy condition $\mu_{EM}\geq|J_{EM}|_{g}$ holds.
If the coupled system of equations \eqref{3.2}, \eqref{4.2}, and \eqref{4.7} admits a solution satisfying the asymptotics \eqref{3.5}, \eqref{3.12}-\eqref{3.14}, and \eqref{4.5}, with a weak IMCF in the sense of \cite{HuiskenIlmanen}, and such that the boundary conditions \eqref{4.1c}, \eqref{4.1d} are valid, then
\begin{equation}\label{5.5}
m\geq\sqrt{\frac{A}{16\pi}}+\sqrt{\frac{\pi}{A}}\mathcal{Q}_{e}^2
\end{equation}
where $A$ is the minimum area required to enclose the boundary. Furthermore,
if equality is achieved then the initial data arise from an embedding into the Reissner-Nordstr\"{o}m-AdS spacetime.
\end{theorem}

\begin{proof}
The inequality \eqref{5.5} follows from the proof of Theorem \ref{thm4.2} by replacing the role of the Hawking mass by the so called charged Hawking mass \cite{DisconziKhuri}, which has appropriate monotonicty properties along IMCF under the charged dominant energy condition. Details may be found in \cite{DisconziKhuri}.

Consider now the case of equality in \eqref{5.5}.
According to \cite{DisconziKhuri} the data set $(M_{3}, g_{3}, E_{3})$ is isometric to the $t=0$ slice of Reissner-Nordstr\"{o}m spacetime, and
\begin{equation}
\widetilde{u} = \sqrt{1 - \frac{2m_{3}}{r} + \left(\frac{\mathcal{Q}_{e}^{3}}{r}\right)^{2}} = \sqrt{1 - \frac{2m}{r}+ \left(\frac{\mathcal{Q}_{e}}{r}\right)^{2}},\quad\quad
E_{3}= \frac{\mathcal{Q}_{e}}{r^{2}}\left(1-\frac{2m}{r}+ \left(\frac{\mathcal{Q}_{e}}{r}\right)^{2}\right)^{-\frac{1}{2}} dr.
\end{equation}
Then $(M_{2}, g_{2}, k_{2},E_2)$ is a graphical totally umbilic slice of the Reissner-Nordstr\"{o}m spacetime. As in the proof of Theorem \ref{thm4.2}, the
uniqueness of such umbilical slices would imply that $\widetilde{f} = \widetilde{f}(r)$ arises from the ODE
\begin{equation} \label{5.12}
\widetilde{f}' = \frac{r}{\left( 1-\frac{2m}{r} + \left(\frac{\mathcal{Q}_{e}}{r}\right)^{2}\right)
\sqrt{1-\frac{2m}{r}+\left(\frac{\mathcal{Q}_{e}}{r}\right)^{2}+r^2}}.
\end{equation}
We then have
\begin{equation} \label{5.13}
u = \widetilde{u}\sqrt{1+ \widetilde{u}^{2}|\widetilde{\nabla} \widetilde{f}|_{g_{2}}^{2}}
= \frac{\widetilde{u}}{\sqrt{1- \widetilde{u}^{2}|\nabla\widetilde{f}|_{g_{3}}^{2}}}
= \sqrt{1-\frac{2m}{r}+ \left(\frac{\mathcal{Q}_{e}}{r}\right)^{2}+r^2},
\end{equation}
which is the warping factor of the Reissner-Nordstr\"{o}m-AdS solution. Moreover
\begin{equation}
E_{1}=E_{2}
= \frac{E_3 - \widetilde{u}^2 E_3(\nabla \widetilde{f})d\widetilde{f}}
{\sqrt{1 - \widetilde{u}^2 |\nabla  \widetilde{f}|^2_{g_{3}}}}
= \frac{\mathcal{Q}_{e}}{r^{2}}\left(1-\frac{2m}{r} + \left(\frac{\mathcal{Q}_{e}}{r} \right)^{2}+r^2\right)^{-\frac{1}{2}} dr,
\end{equation}
and it is straightforward to see that $(M_{1}, g_{1}=g_{2}=g_{3}- \widetilde{u}^{2}d \widetilde{f}^2)$ is then isometric to the $t=0$ slice of Reissner-Nordstr\"{o}m-AdS. Therefore $(M, g, k,E)$ indeed emerges from an embedding into the Reissner-Nordstr\"{o}m-AdS spacetime given by the solution of the generalized Jang equation \eqref{3.2}, after imitating the case of equality arguments in \cite{BrayKhuri1,BrayKhuri2,DisconziKhuri}.
\end{proof}

\section{The Positive Mass Theorem with Charge}
\label{sec6} \setcounter{equation}{0}
\setcounter{section}{6}

%[[Nonexistence of MP-AdS likely due to lack of harmonic map structure, ie. no superposition %of static vacuum solutions. Physically the lack may be due to the fact that a negative %cosmological constant enhances gravity so it is hard to balance the electromagnetic and %gravitational forces.]]

%[[Need $H^2-Trk^2=4$ condition on initial data in order to get horizon in $M_2$. Then
%need $\tilde{f}$ to satisfy asymptotics as in \cite{KhuriWeinstein} to get cylindrical %ends.]]

%Extreme RN-AdS has a naked singularity and no horizon.\cite{HristovToldaVandoren}

Consider the setting of the previous section, without the divergence free assumption on
the electric field $E$. The three deformations of the previous section, leading to the asymptotically flat data $(M_{3},g_{3}, E_{3})$, will be employed here with the following
modifications. Instead of choosing boundary behavior for $\widetilde{f}$ to give rise to a minimal boundary for $M_3$ as in \eqref{4.1d}, here the boundary behavior of $\widetilde{f}$ will be prescribed to achieve:
\begin{equation}\label{bdry}
\text{$M_3$ has asymptotically cylindrical ends over the horizons of $M_2$.}
\end{equation}
This type of blow-up behavior for solutions of the classical Jang equation originated in the proof of the positive mass theorem by Schoen and Yau \cite{SchoenYau} in the asymptotically flat case, and has been studied and used in the context of the generalized Jang equation in \cite{ChaKhuriSakovich,HanKhuri1,KhuriWeinstein}. In addition, $\widetilde{u}$ will be chosen differently than in previous sections while $u$ will still be constructed from $\widetilde{u}$ according to \eqref{4.7}.

We now describe how to choose $\widetilde{u}$ following \cite{ChaKhuriSakovich}.
Let $\mathcal{S}$ denote the $SL(2,\mathbb{C})$ spinor bundle over $M_3$ with the Einstein-Maxwell \cite{GHHP} spin connection
\begin{equation}
   \nabla_{e_{i}}= \hat{\nabla}_{e_i}
   -\frac{1}{2} \, E_3\cdot e_{i}\cdot e_{0}\cdot,
\end{equation}
where $\cdot$ indicates Clifford multiplication, $\hat{\nabla}$ is the $g_{3}$ metric compatible connection on $\mathcal{S}$, $(e_1,e_2,e_3)$ is an orthonormal frame field for $M_3$ and $e_0$ is the unit normal to $M_3$ in the Lorentzian warped product $(M_3\times\mathbb{R}, g_3-\widetilde{u}^2 dt^2)$. Let $\psi_0$ be a given constant spinor of modulus 1, and consider a harmonic spinor
\begin{equation}\label{dirac}
  \slashed{D}\psi = \sum_{i=1}^3 e_i\cdot
  \nabla_{e_i} \psi=0,
\end{equation}
satisfying the asymptotic conditions
\begin{equation}\label{dirac1}
\psi\rightarrow\psi_0\text{ }\text{ as }\text{ }r\rightarrow\infty,\quad\quad\quad
\psi\rightarrow 0\text{ }\text{ along the cylindrical ends of $M_3$.}
\end{equation}
The Dirac equation is coupled to the generalized Jang equation \eqref{4.2} through the choice
\begin{equation}\label{psi}
\widetilde{u}=|\psi|^{2},\quad\quad\quad
u=\widetilde{u}\sqrt{1+\widetilde{u}^2|\widetilde{\nabla}\widetilde{f}|_{g_2}^2}.
\end{equation}
According to Lemma 5.1 of \cite{ChaKhuriSakovich}, under mild conditions on the asymptotics of $(M_{3}, g_{3}, E_{3})$, $\widetilde{u}$ defined by \eqref{psi} will satisfy the asymptotic condition \eqref{4.5}.

\begin{theorem}\label{thm6.1}
Let $(M, g, k, E)$ be a $3$-dimensional, asymptotically AdS hyperbolic initial data set
for the Einstein-Maxwell equations with an outermost mean convex boundary satisfying \eqref{4.1b}, and such that the charged dominant energy condition \eqref{10} holds. If the coupled system of equations \eqref{3.2}, \eqref{4.2}, \eqref{dirac}, and \eqref{psi} admits a solution satisfying the asymptotics \eqref{3.5}, \eqref{3.12}-\eqref{3.14}, and \eqref{4.5}, and such that the boundary conditions \eqref{4.1c}, \eqref{bdry}, \eqref{dirac1} are valid, then
\begin{equation}\label{masscharge}
m\geq|\mathcal{Q}_{e}|.
\end{equation}
Furthermore, if equality is achieved and $\partial M$ has only one component
then the initial data arise from an embedding into the BPS Reissner-Nordstr\"{o}m-AdS spacetime.
\end{theorem}

\begin{remark}
In the asymptotically flat and asymptotically hyperboloidal setting, the analogous BPS bound \eqref{masscharge} is saturated only for the Majumdar-Papapetrou black holes \cite{ChaKhuriSakovich,KhuriWeinstein}, whose horizon may have multiple components. The current theorem only treats the rigidity statement for a single component boundary since it is not clear what spacetime should serve as the ground state when more than one component is present. In particular, there are no known analogues of the Majumdar-Papapetrou solutions in the AdS regime. This is perhaps due to the lack of solution generating techniques, as the axisymmetric stationary (or static) electrovacuum Einstein equations with nonvanishing cosmological constant do not reduce to a sigma model.
\end{remark}

\begin{proof}
In \cite{ChaKhuriSakovich}, the inequality \eqref{masscharge} is treated for asymptotically hyperboloidal data. Thus, with only minor modifications we may apply those arguments to $(M_{2}, g_{2}, k_{2}, E_2)$ to find $m_{2} \geq |Q^{2}_{e}|$. The only difference arises from the scalar curvature formulas \eqref{4.6} above, and (3.3) in \cite{ChaKhuriSakovich}.
In the present work an extra divergence term appears in the scalar curvature expression, which yields an extra boundary integral.
However, as explained in \cite{ChaKhuriSakovich} the choice \eqref{psi}
ensures that $\widetilde{u}$ (and $u$) exhibit exponential decay along the cylindrical ends
of $M_3$, so that the associated boundary integrals vanish.
The desired conclusion now follows from $m=m_{2}$ (Lemmas \ref{lemma3.1} and \ref{lemma3.2}) and $Q=Q_{2}$ \eqref{q}.

Consider now the case of equality in \eqref{masscharge} when $\partial M$ has one component. According to \cite{ChaKhuriSakovich} and \cite{KhuriWeinstein}, $(M_3,g_3,E_3)$ is then isometric to the $t=0$ slice of a Majumdar-Papapetrou spacetime; since the horizon has only one component this is equivalent to extreme Reissner-Nordstr\"{o}m. The proof of Theorem \ref{thm5.1} now shows that $(M,g,k,E)$ arises from an embedding into the Reissner-Nordstr\"{o}m-AdS solution with $m=\mathcal{Q}_e$, which we refer to as a BPS state.
\end{proof}

In the asymptotically flat and asymptotically hyperboloidal case the BPS bound
agrees with \eqref{masscharge}, and is saturated by a proper black hole, namely an
extreme Reissner-Nordstr\"{o}m black hole. However, in the AdS case the BPS ground state
(RN-AdS with $m=\mathcal{Q}_e$) has a naked singularity. This suggests the possibility of an alternative mass-charge inequality in the AdS regime which is saturated by a proper black hole. The most likely candidate is the extreme RN-AdS solution, in which the inner and outer horizons coincide. Let $m_{ext}$ denote the mass of such a black hole for a given charge, then we conjecture that the following mass-charge inequality holds, under the charged dominant energy condition, for asymptotically AdS hyperbolic initial data
\begin{equation}\label{extreme}
m\geq m_{ext}=\frac{1}{3\sqrt{6}}\left(\sqrt{1+12\mathcal{Q}_{e}^2}+2\right)
\left(\sqrt{1+12\mathcal{Q}_{e}^2}-1\right)^{1/2}.
\end{equation}
Note that $m_{ext}\geq m_{bps}:=|\mathcal{Q}_e|$ with equality if and only if $\mathcal{Q}_e=0$. Therefore \eqref{extreme} is a stronger inequality than \eqref{masscharge}.

\section{The Mass-Angular Momentum Inequality}
\label{sec7} \setcounter{equation}{0}
\setcounter{section}{7}

Based on the Jang-type deformations of the previous sections we propose here a new
relation between mass and angular momentum for asymptotically AdS hyperbolic initial
data, namely
\begin{equation}\label{7.0}
m\geq\sqrt{|\mathcal{J}|}.
\end{equation}
This inequality is respectively stronger/weaker for small/large angular momentum, when compared with the BPS bound
\begin{equation}\label{7.00}
m\geq|\mathcal{J}|
\end{equation}
established in \cite{ChruscielMaertenTod}. We conjecture that the optimal mass-angular momentum inequality in this setting is given by
\begin{equation}\label{7.000}
m\geq \frac{1}{3\sqrt{6}}\left[\sqrt{\left(1+\frac{\mathcal{J}^2}{m^2}\right)^2+
\frac{12\mathcal{J}^2}{m^2}}+2\left(1+\frac{\mathcal{J}^2}{m^2}\right)\right]
\left[\sqrt{\left(1+\frac{\mathcal{J}^2}{m^2}\right)^2+\frac{12\mathcal{J}^2}{m^2}}
-\left(1+\frac{\mathcal{J}^2}{m^2}\right)\right]^{1/2},
\end{equation}
which is saturated for extreme Kerr-AdS black holes \cite{CaldarelliKlemm}. If $m_{ext}$ denotes the right-hand side of \eqref{7.000}, then a computation shows that $m_{ext}> |\mathcal{J}|/m$ and $m_{ext}>|\mathcal{J}|$ unless $\mathcal{J}=0$. Therefore \eqref{7.000} implies both \eqref{7.0} and \eqref{7.00}. However it is not expected \cite{HristovToldaVandoren} that any proper black hole solution saturates the bounds \eqref{7.0} or \eqref{7.00}, since within the Kerr-AdS family such a state would result in a naked singularity.

In contrast to previous sections, here the initial data $(M,g,k)$ shall have two
asymptotic ends. One end $M_{end}^{+}$, from which the mass $m$ arises, is designated asymptotically AdS hyperbolic, whereas the other end $M_{end}^{-}$ is either asymptotically AdS hyperbolic or asymptotically cylindrical. In the asymptotically cylindrical case the end is diffeomorphic to $S^2\times [r_0,\infty)$ and with the coordinates given by this diffeomorphism, as $r\rightarrow 0$ the fall-off conditions are
\begin{equation}\label{7.1b}
\begin{split}
&g_{rr}= r^{-2} + o_{1}(r^{-\frac{3}{2}}), \text{ }\text{ }\text{ }\text{ }
g_{r \alpha}=o_{1}(r^{-\frac{1}{2}}),\text{ }\text{ }\text{ }\text{ }
g_{\alpha \beta}= \sigma_{\alpha \beta} + o_{1}(r^{\frac{1}{2}}), \\
&k_{rr}=O(r^{-2}), \text{ }\text{ }\text{ } \text{ }\text{ }\text{ }\text{ }\text{ }
k_{r \alpha}=O(r^{-1}), \text{ }\text{ }\text{ }\text{ }\text{ } \text{ }\text{ }
k_{\alpha \beta}=O(1),
\end{split}
\end{equation}
where $\sigma$ is the round metric and $\alpha$, $\beta$ are coordinates on the sphere. The asymptotic conditions for $k$ are chosen to facilitate the reduction argument.
Moreover, for later reference in this section and the next, we record the modified decay rates to be used in asymptotically flat ends, namely
\begin{equation}\label{7.1}
g_{ij} - \delta_{ij} = o_{2}(r^{-\frac{1}{2}}), \quad\quad\text{ }\text{ }\text{ }k_{ij}=O_{1}(r^{-\frac{5}{2}-\varepsilon})\text{ }\text{ }\text{ }\text{ }\text{ }\text{ as }\text{ }\text{ }\text{ }\text{ }\text{ }r \to \infty
\end{equation}
for some $\varepsilon>0$.

It will be assumed that $M$ is simply connected and that the data are
are axisymmetric. The later means
that a subgroup of the group of isometries is isomorphic to $U(1)$, and that all quantities associated with the initial data are invariant under this $U(1)$ action. In particular, if $\eta=\partial_{\phi}$ denotes the Killing field which generates the symmetry, then
\begin{equation}\label{7.2}
\mathfrak{L}_{\eta}g=\mathfrak{L}_{\eta}k=0,
\end{equation}
where $\mathfrak{L}_{\eta}$ denotes Lie differentiation. Axisymmetry allows for simple definition of angular momentum via the following Komar integral
\begin{equation}\label{7.3}
\mathcal{J}=\frac{1}{8\pi}\int_{S}(k_{ij}-(Tr_{g} k)g_{ij})\nu^{i}_{g}\eta^{j},
\end{equation}
where $\nu_{g}$ is the unit outer normal to any surface $S$ which is homologous to
the `sphere at infinity' associated with $M_{end}^{+}$. This definition is
well-defined (independent of the choice of $S$) as long as
\begin{equation}\label{7.4}
J_i\eta^{i}=0.
\end{equation}

We seek three deformations of the initial data such that $M_{end}^{+}$ is
transformed into an asymptotically flat end, and $M_{end}^{-}$ is transformed into
an asymptotically flat/asymptotically cylindrical end if it was originally
asymptotically AdS hyperbolic/asymptotically cylindrical.
Furthermore, the topology of $M$ as well as physical quantities such as mass and angular momentum will be preserved.
In the first deformation a new data set $(M_{1}, g_{1}, k_{1})$ is constructed based largely on the ideas of \cite{ChaKhuri1}. Consider a graph $M_{1}=\{t=f(x)\}$ embedded in the warped product stationary 4-manifold
\begin{equation}
(M \times \mathbb{R}, g + 2Y_{i}dx^{i}dt+(u^{2}-|Y|_{g_{1}}^{2}) dt^2),
\end{equation}
with induced metric
\begin{equation}
(g_{1})_{ij}=g_{ij}+f_{i}Y_{j}+f_{j}Y_{i}+(u^{2}-|Y|_{g_{1}}^{2}) f_{i}f_{j}.
\end{equation}
The new extrinsic curvature is defined to be the second fundamental form of the $t=0$ slice in the dual Lorentzian setting
\begin{equation}
(M_{1} \times \mathbb{R}, g_{1} - 2Y_{i}dx^{i}dt-(u^2 - |Y|^{2}_{g_{1}}) dt^2),
\end{equation}
and is given by
\begin{equation}\label{7.7}
(k_{1})_{ij}=\frac{1}{2u}\left(\overline{\nabla}_{i}Y_{j}+\overline{\nabla}_{j}Y_{i}\right),
\end{equation}
where $\overline{\nabla}$ is the Levi-Civita connection with respect to $g_{1}$.
Note also that $M=\{t=f(x)\}$ is embedded in this spacetime with induced metric $g$. Moreover, motivated by the structure of the Kerr-AdS spacetime, it will be assumed that $Y$ has a single component in the Killing direction
\begin{equation}\label{7.8}
Y^{i}\partial_{i}:=(g_{1})^{ij}Y_{j}\partial_{i}=Y^{\phi}\partial_{\phi}.
\end{equation}
This condition also guarantees that $g_1$ is a Riemannian metric \cite{ChaKhuri1}.
Thus, the first deformation is characterized by three functions $(u,Y^{\phi},f)$, which are axisymmetric
\begin{equation}\label{7.6}
\mathfrak{L}_{\eta}u=\mathfrak{L}_{\eta}Y^{\phi}=\mathfrak{L}_{\eta}f=0,
\end{equation}
and will be chosen appropriately below.

In order to have a well-defined angular momentum after deformation, \eqref{7.4} must continue to hold
\begin{equation}\label{7.9}
\operatorname{div}_{g_{1}}k_{1}(\eta)=0.
\end{equation}
This is a linear elliptic equation for $Y^{\phi}$, if $u$ and $f$ are independent of $Y^{\phi}$.  The total angular momentum will be preserved if the following asymptotics
hold in the designated AdS hyperbolic end
\begin{equation}\label{7.10}
Y^{\phi}=-\frac{2\mathcal{J}(\theta)}{r^{3}}+O_{3}(r^{-4})\text{ }\text{ }\text{ as }\text{ }\text{ }r\rightarrow\infty,
\end{equation}
for a function $\mathcal{J}(\theta)$ on $S^{2}$ such that
\begin{equation}\label{7.10'}
\int_{0}^{\pi}  \mathcal{J}(\theta)\sin^{3} \theta d\theta = \frac{4}{3}\mathcal{J}.
\end{equation}
As in \cite{ChaKhuri1,ChaKhuriSakovich}, asymptotics in the other end are given by
\begin{equation}\label{7.11}
Y^{\phi}=O_{3}(r^{5})\text{ }\text{ }\text{ in asymptotically AdS hyperbolic }M_{end}^{-},
\end{equation}
\begin{equation}\label{7.12}
Y^{\phi}=\mathcal{Y}+O_{3}(r)\text{ }\text{ }\text{ in asymptotically cylindrical }M_{end}^{-},
\end{equation}
where $\mathcal{Y}$ is a constant determined by the data. Here $r$ arises from Kelvin inversion of the coordinates in \eqref{1} and \eqref{2} in the asymptotically AdS case.
%The existence of solutions to \eqref{7.9} satisfying the desired asymptotics will be shown in %Theorem \ref{thm7.6} below, in the case that $u$ and $f$ are independent of $Y^{\phi}$.

Next we choose $f$ to satisfy the Jang-type equation
\begin{equation}\label{7.13}
g^{ij}\left(\frac{u\nabla_{ij}f+u_{i}f_{j}+u_{j}f_{i}}{\sqrt{1+u^{2}|\nabla f|_{g}^{2}}}-k_{ij}\right)=0\text{ }\text{ }\text{ }\text{ }\text{ }\Leftrightarrow\text{ }\text{ }\text{ }\text{ }\text{ }
\operatorname{div}_{g}(u^{2}\nabla f)=u(Tr_{g}k)\sqrt{1+u^{2}|\nabla f|_{g}^{2}},
\end{equation}
with fall-off as $r\rightarrow\infty$ in the designated asymptotically AdS hyperbolic end given by
\begin{equation}\label{7.14}
f =  o_{4}\left(r^{-3}\right),
\end{equation}
and asymptotics as $r\rightarrow 0$ satisfying
\begin{equation}\label{7.15}
f = O_{4}(r)\text{ }\text{ }\text{ in asymptotically AdS hyperbolic }M_{end}^{-},
\end{equation}
\begin{equation}\label{7.16}
\sum_{p=1}^{4}|\nabla^{p} f|_{g}=o(r^{\frac{1}{2}})\text{ }\text{ }\text{ in asymptotically cylindrical }M_{end}^{-}.
\end{equation}
Existence of solutions to \eqref{7.13} satisfying the desired asymptotics will be shown
in Theorem \ref{thm7.5} below, in the case that $u$ is independent of $f$.
The purpose of the Jang-type equation \eqref{7.13} is to guarantee weak nonnegativity
of matter energy density for the new data. More precisely it is shown in \cite{ChaKhuri1} that
\begin{equation}\label{7.17}
2\mu_{1}=R_{1}-|k_{1}|_{g_{1}}^{2}-2\Lambda= 2(\mu-J(v))+|k-\pi|_{g}^{2}+2u^{-1}\operatorname{div}_{g_{1}}(uQ),
\end{equation}
where
\begin{equation}\label{7.18}
\pi_{ij}=\frac{u\nabla_{ij}f+u_{i}f_{j}+u_{j}f_{i}+\frac{1}{2u}(g_{i\phi}Y^{\phi}_{,j}+g_{j\phi}Y^{\phi}_{,i})}{\sqrt{1+u^{2}|\nabla f|_{g}^{2}}}
\end{equation}
denotes extrinsic curvature of the graph in the Lorentzian setting,
\begin{equation}\label{7.19}
v^{i}=\frac{uf^{i}}{\sqrt{1+u^{2}|\nabla f|_{g}^{2}}},\text{ }\text{ }\text{ }\text{ }\text{ }\text{ }\quad\text{ }w^{i}=\frac{uf^{i}+u^{-1}Y^{\phi}\partial_{\phi}}{\sqrt{1+u^{2}|\nabla f|_{g}^{2}}},
\end{equation}
and
\begin{equation}\label{7.20}
Q_{i}=g_{1}^{jl}Y_{l}\overline{\nabla}_{ij}f-ug_{1}^{jl}f_{l}(k_{1})_{ij}+w^{j}(k-\pi)_{ij}+uf_{i}w^{l}w^{j}(k-\pi)_{lj}\sqrt{1+u^{2}|\nabla f|_{g}^{2}}.
\end{equation}
If the dominant energy condition \eqref{5} is valid, then $\mu_{1} \geq 0$ holds weakly in the sense that this quantity differs from a nonnegative function by a divergence term.

It remains to choose $u$. This will require knowledge of further deformations, and thus will be explained at the end. Here we simply note the appropriate asymptotics. As $r\rightarrow\infty$ in $M_{end}^{+}$ the desired expansion is
\begin{equation}\label{7.21}
u= \sqrt{1+r^{2}}+ u_{0} + O_{3}(r^{-1})
\end{equation}
for a constant $u_0$, and as $r\rightarrow 0$ the required asymptotics are given by
\begin{equation}\label{7.22}
u=r+o_{3}(r^{\frac{3}{2}})\text{ }\text{ }\text{ }\text{ in asymptotically AdS hyperbolic }M_{end}^{-},
\end{equation}
\begin{equation}\label{7.23}
u=r+o_{3}(r^{\frac{3}{2}})\text{ }\text{ }\text{ }\text{ in asymptotically cylindrical }M_{end}^{-}.
\end{equation}

\begin{lemma}\label{lemma7.1}
The initial data set $(M_1,g_1,k_1)$ is axially symmetric and maximal $Tr_{g_{1}}k_{1}=0$, and preserves the asymptotic geometry of $(M,g,k)$ as well as the mass and angular momentum $m_1=m$,
$\mathcal{J}_{1}=\mathcal{J}$.
%If $M_{end}^{+}$ is asymptotically AdS hyperbolic and \eqref{7.21}, \eqref{7.10}, and %\eqref{7.14} are satisfied then the data set $((M_{1})^{+}_{end}, g_{1}, k_{1})$ is %asymptotically AdS hyperbolic and maximal. Similarly, if $M_{end}^{-}$ is asymptotically %%AdS hyperbolic and \eqref{7.22}, \eqref{7.11}, and \eqref{7.15} are satisfied then the %data set $((M_{1})^{-}_{end}, g_{1}, k_{1})$ is also asymptotically AdS hyperbolic and %maximal. Moreover, the mass and the angular momentum of the new data are given by %$m_{1}=m$, $\mathcal{J}_{1}=\mathcal{J}$.
\end{lemma}

\begin{proof}
Maximality is established in \cite{ChaKhuri1}, and axisymmetry is clear. Furthermore,
it follows from Lemma \ref{lemma3.1} that $((M_{1})^{+}_{end}, g_{1})$ is asymptotically AdS hyperbolic and that the mass is given by $m_{1}=m$.
A computation \cite[Lemma 2.1]{ChaKhuri1} shows that
\begin{equation} \label{7.24}
2u (k_{1})_{ij} = (g_{1})_{\phi i}\partial_{j}Y^{\phi} + (g_{1})_{\phi j}\partial_{i}Y^{\phi},
\end{equation}
which together with \eqref{7.10} and \eqref{7.21} yields
\begin{equation} \label{7.25}
(k_{1})_{rr}= O_{2}(r^{-8}), \text{ }\text{ }\text{ }\text{ }\text{ }\text{ }\text{ }\text{ }\text{ }\text{ }
(k_{1})_{r\alpha} = O_{2}(r^{-3}),\text{ }\text{ }\text{ }\text{ }\text{ }\text{ }\text{ }\text{ }\text{ }\text{ }
(k_{1})_{\alpha \beta} = O_{2}(r^{-2}),
\end{equation}
as $r\rightarrow\infty$. This shows that $k_{1}$ satisfies \eqref{3}.
In order to calculate the angular momentum use \eqref{1}, \eqref{2}, \eqref{7.10}, \eqref{7.10'}, \eqref{7.21}, and \eqref{7.24} to find
\begin{equation}\label{7.26}
\mathcal{J}_{1}
=\frac{1}{8\pi}\int_{\overline{S}_{\infty}}k_{1}(\nu_{g_{1}}, \eta)
= \frac{1}{16\pi}\int_{\overline{S}_{\infty}} \frac{g_{\phi \phi} \partial_{r}Y^{\phi}}{u|\partial_{r}|_{g_{1}}}
= \frac{3}{4} \int_{0}^{\pi}  \mathcal{J}(\theta)\sin^{3} \theta d\theta
 = \mathcal{J}.
\end{equation}

Consider now the case that $(M^{-}_{end}, g, k)$ is asymptotically AdS hyperbolic. Kelvin inversion then implies the asymptotics
\begin{equation} \label{7.230}
g_{rr}=\frac{1}{r^{2}} - 1 + \mathbf{m}^{r}r + O_{3}(r^{2}),\text{ }\text{ }\text{ }\text{ }\text{ }\text{ }\text{ }\text{ }
g_{r \alpha} = O_{3}(r),\text{ }\text{ }\text{ }\text{ }\text{ }\text{ }\text{ }\text{ }
g_{\alpha \beta} = \frac{\sigma_{\alpha \beta}}{r^{2}} + \mathbf{m}^{g}_{\alpha \beta} r + O_{3}(r^{2}),
\end{equation}
as $r\rightarrow 0$.  This together with \eqref{7.11}, \eqref{7.15}, and \eqref{7.22} yields
\begin{equation} \label{7.232}
\begin{split}
(g_{1})_{rr}
&= g_{rr} + 2 Y_{r}f_{r} + \left(u^{2} - |Y|_{g_{1}}^{2} \right)f_{r}^{2} \\
&= g_{rr} + 2 g_{r \phi}Y^{\phi}f_{r} + \left(u^{2} + (Y^{\phi})^{2}g_{\phi \phi} \right)f_{r}^{2} \\
&= \frac{1}{r^{2}} - 1 + \mathbf{m}^{r}r + O_{3}(r^{2}),
\end{split}
\end{equation}
\begin{equation} \label{7.233}
\begin{split}
(g_{1})_{r \alpha}
&= g_{r \alpha} + g_{\alpha \phi}Y^{\phi}f_{r} + g_{r \phi}Y^{\phi}f_{\alpha} + \left(u^{2} + (Y^{\phi})^{2}g_{\phi \phi} \right)f_{r}f_{\alpha} \\
&= O_{3}(r),
\end{split}
\end{equation}
\begin{equation} \label{7.234}
\begin{split}
(g_{1})_{\alpha \beta}
&= g_{\alpha \beta} + g_{\beta \phi}Y^{\phi}f_{\alpha} + g_{\alpha \phi}Y^{\phi}f_{\beta} + \left(u^{2} + (Y^{\phi})^{2}g_{\phi \phi} \right)f_{\alpha}f_{\beta} \\
&=\frac{\sigma_{\alpha \beta}}{r^{2}} + \mathbf{m}^{g}_{\alpha \beta} r + O_{3}(r^{2}).
\end{split}
\end{equation}
This shows that $((M_{1})^{-}_{end}, g_{1})$ satisfies \eqref{1} and \eqref{2}. In addition, from
\eqref{7.24} we find that
\begin{equation} \label{7.251}
(k_{1})_{rr}= O_{2}(r^{4}), \text{ }\text{ }\text{ }\text{ }\text{ }\text{ }\text{ }\text{ }\text{ }\text{ }
(k_{1})_{r\alpha} = O_{2}(r),\text{ }\text{ }\text{ }\text{ }\text{ }\text{ }\text{ }\text{ }\text{ }\text{ }
(k_{1})_{\alpha \beta} = O_{2}(r^{2}),
\end{equation}
which shows that $k_1$ satisfies \eqref{3} on $(M_{1})^{-}_{end}$. Similar considerations
are used in the asymptotically cylindrical case.
\end{proof}

We will now proceed to the second deformation which is based on a procedure introduced in \cite{ChruscielTod}. The goal is to construct a constant mean curvature (CMC) asymptotically hyperboloidal data set, denoted $(M_{2}, g_{2}, k_{2})$, from the maximal asymptotically AdS data $(M_{1}, g_{1}, k_{1})$. In order to accomplish this define
\begin{equation} \label{7.27}
(M_{2}, g_{2}) \equiv (M_{1}, g_{1}), \text{ }\text{ }\text{ }\text{ }\text{ }\text{ }\text{ }\text{ }\text{ }\text{ }\text{ }\text{ }\text{ }\text{ }
(k_{2})_{ij} = (k_{1})_{ij} + \sqrt{\frac{-\Lambda}{3}} (g_{1})_{ij}.
\end{equation}
Observe that
\begin{equation}\label{7.31}
Tr_{g_{2}} k_{2} = Tr_{g_{1}} k_{1} + 3 = -\Lambda, \quad\quad\quad
|k_{2}|_{g_{2}}^{2} = |k_{1}|_{g_{1}}^{2} + 2 Tr_{g_{1}} k_{1} +  3 = |k_{1}|_{g_{1}}^{2} - \Lambda,
\end{equation}
and therefore
\begin{equation} \label{7.32}
\begin{split}
& 2\mu_{2} = R_{2} + (Tr_{g_{2}} k_{2})^{2} - |k_{2}|_{g_{2}}^{2}
= R_{1} - |k_{1}|_{g_{1}}^{2} -2 \Lambda = 2 \mu_{1} ,  \\
& J_{2} = \operatorname{div}_{g_{2}}(k_{2}- (Tr_{g_{2}} k_{2}) g_{2}) = J_{1}.
\end{split}
\end{equation}
In particular, this shows that the condition for a well-defined angular momentum
is satisfied, namely
\begin{equation} \label{7.33}
J_{2} (\eta) = J_{1} (\eta) = 0,
\end{equation}
according to \eqref{7.9}

\begin{lemma}\label{lemma7.2}
The initial data set $(M_2,g_2,k_2)$ is CMC and axially symmetric with two ends, one designated asymptotically hyperboloidal $(M_{2})_{end}^{+}$ and the other $(M_{2})_{end}^{-}$ either asymptotically hyperboloidal or asymptotically cylindrical depending on whether $(M_{1})_{end}^{-}$ is asymptotically AdS hyperbolic or asymptotically cylindrical. Furthermore the mass and angular momentum are preserved, that is $m_{2}=m_{1}$ and $\mathcal{J}_{2}=\mathcal{J}_{1}$.
\end{lemma}

\begin{proof}
It is clear that the new data set is CMC and axisymmetric. Since $\Lambda=-3$,
we have $k_2=k_1+g_1$ so that the decay rates for $g_{1}$ in \eqref{3.7} and for $k_{1}$ in \eqref{7.25} yield
\begin{equation} \label{7.29}
 (k_{2})_{rr} = \frac{1}{1+r^{2}} + \frac{\mathbf{m}^{r} }{r^{5}} + O_{2}(r^{-6}),  \text{ }\text{ }\text{ }\text{ }\text{ }\text{ }
 (k_{2})_{r\alpha} = O_{2}(r^{-3}),
 \text{ }\text{ }\text{ }\text{ }\text{ }\text{ }
 (k_{2})_{\alpha \beta} = r^{2} \sigma_{\alpha \beta}+ \frac{\mathbf{m}^{g}_{\alpha\beta} }{r}+ O_{2}(r^{-2}),
\end{equation}
as $r\rightarrow\infty$. This shows that $(M_{2})_{end}^{+}$ is asymptotically hyperboloidal. The mass is then given by
\begin{equation} \label{7.30}
m_{2}
= \frac{1}{16\pi} \int_{S^{2}} \left[ Tr_{\sigma}\left( \mathbf{m}^{g}_{2} + 2\mathbf{m}^{k}_{2}\right) + 2\mathbf{m}^{r}_{2} \right]
= \frac{1}{16\pi} \int_{S^{2}} 3 Tr_{\sigma} \mathbf{m}^{g} + 2\mathbf{m}^{r} =m_{1}.
\end{equation}
In addition, angular momentum is preserved $\mathcal{J}_{2}=\mathcal{J}_{1}$ since the normal vector of an axisymmetric $2$-surface is perpendicular to the killing vector $\eta$. Lastly, the desired asymptotics for $(M_2)_{end}^{-}$
follow from \eqref{7.232}-\eqref{7.234} and \eqref{7.251} in the asymptotically hyperboloidal case; similar considerations may be used in the asymptotically cylindrical case.
\end{proof}

The third and final deformation is similar to the first, however here the asymptotically
hyperboloidal ends of $M_2$ will be transformed into asymptotically flat ends.
Consider a graph $M_{3}=\{t=\widetilde{f}(x)\}$ embedded in the warped product 4-manifold \begin{equation}
(M_{2} \times \mathbb{R}, g_{2} + 2\widetilde{Y}_{i}dx^{i}dt+(\widetilde{u}^{2}-|\widetilde{Y}|_{g_{3}}^{2}) dt^2).
\end{equation}
The new data $(M_{3}, g_{3}, k_{3})$ arise from the induced metric and second fundamental form of the $t=0$ slice in the dual Lorentzian setting, that is
\begin{equation}\label{7.34}
(g_{3})_{ij}=(g_{2})_{ij}+ \widetilde{Y}_{i}\widetilde{f}_{j} + \widetilde{Y}_{j}\widetilde{f}_{i} + (\widetilde{u}^{2}-|\widetilde{Y}|_{g_{3}}^{2}) \widetilde{f}_{i}\widetilde{f}_{j},\text{ }\text{ }\text{ }\text{ }\text{ }\text{ }\text{ }
(k_{3})_{ij}=\frac{1}{2\widetilde{u}}\left(\hat{\nabla}_{i}\widetilde{Y}_{j}+\hat{\nabla}_{j}\widetilde{Y}_{i}\right),
\end{equation}
where $\hat{\nabla}$ is the Levi-Civita connection associated to the metric $g_{3}$. As in the first deformation $\widetilde{Y}$ is set to have a single component in the Killing direction
\begin{equation}\label{7.35}
\widetilde{Y}^{i}\partial_{i}:=(g_{3})^{ij}\widetilde{Y}_{j}\partial_{i}=\widetilde{Y}^{\phi}\partial_{\phi}.
\end{equation}
Thus, the third deformation is characterized by three functions $(\widetilde{u}, \widetilde{Y}^{\phi}, \widetilde{f})$, which are axisymmetric
\begin{equation}
\mathfrak{L}_{\eta}\widetilde{u}=\mathfrak{L}_{\eta}\widetilde{Y}^{\phi}
=\mathfrak{L}_{\eta}\widetilde{f}=0,
\end{equation}
and will be chosen appropriately below.

In analogy with \eqref{7.9} the following equation will be imposed which allows for the
existence of a twist potential and well-defined angular momentum
\begin{equation}\label{7.36}
\operatorname{div}_{g_{3}}k_{3}(\eta)=0.
\end{equation}
This is a linear elliptic equation for $\widetilde{Y}^{\phi}$ if the functions $u$, $\widetilde{u}$, $Y^{\phi}$, $f$, and $\widetilde{f}$ are fixed. Existence and uniqueness in this case has been treated in \cite{ChaKhuri1} and \cite{ChaKhuriSakovich} under the following asymptotics
\begin{equation}\label{7.360}
\widetilde{Y}^{\phi}=-\frac{2\mathcal{J}_{2}}{r^{3}}+o_{2}(r^{-\frac{7}{2}})\text{ }\text{ }\text{ as }\text{ }\text{ }r\rightarrow\infty\text{ }\text{ }\text{ in }\text{ }\text{ }(M_2)_{end}^{+},
\end{equation}
\begin{equation}\label{7.361}
\widetilde{Y}^{\phi}=O_{2}(r^{5})\text{ }\text{ }\text{ as }\text{ }\text{ }r\rightarrow 0\text{ }\text{ }\text{ in asymptotically hyperboloidal }(M_{2})_{end}^{-},
\end{equation}
and
\begin{equation}\label{7.362}
\widetilde{Y}^{\phi}=O_{2}(r)\text{ }\text{ }\text{ as }\text{ }\text{ }r\rightarrow 0\text{ }\text{ }\text{ in asymptotically cylindrical }(M_{2})_{end}^{-}.
\end{equation}
%where $\widetilde{\mathcal{Y}}$ is a constant determined by the data.

Next, choose $\widetilde{f}$ to satisfy the following Jang-type equation (see \eqref{7.13})
\begin{equation}\label{7.37}
g_{2}^{ij}\left(\frac{\widetilde{u} \widetilde{\nabla}_{ij} \widetilde{f}
+ \widetilde{u}_{i}\widetilde{f}_{j}
+ \widetilde{u}_{j}\widetilde{f}_{i}}{\sqrt{1+\widetilde{u}^{2}|\widetilde{\nabla}  \widetilde{f}|_{g_{2}}^{2}}}- (k_{2})_{ij}\right)=0,
\end{equation}
where $\widetilde{\nabla}$ is the Levi-Civita connection associated with the metric $g_{2}$. According to \cite{ChaKhuri1} this equation imparts an advantageous structure
to the matter energy density. More precsiely
\begin{equation}\label{7.38}
2\mu_{3} = R_{3}-|k_{3}|_{g_{3}}^{2} = 2(\mu_{2}-J_{2}(\widetilde{v}))
+|k_{2}- \pi_{2}|_{g_{2}}^{2}
+2\widetilde{u}^{-1}\operatorname{div}_{g_{3}}(\widetilde{u}\widetilde{Q}),
\end{equation}
where
\begin{equation}\label{7.39}
(\pi_{2})_{ij}=\frac{ \widetilde{u} \widetilde{\nabla}_{ij}f
+\widetilde{u}_{i}\widetilde{f}_{j}
+\widetilde{u}_{j}\widetilde{f}_{i}
+\frac{1}{2\widetilde{u}}\left((g_{2})_{i\phi}\widetilde{Y}^{\phi}_{,j}+(g_{2})_{j\phi}\widetilde{Y}^{\phi}_{,i}\right)}{\sqrt{1+\widetilde{u}^{2}|\widetilde{\nabla} \widetilde{f}|_{g_{2}}^{2}}}
\end{equation}
is the second fundamental form of the graph in the Lorentzian setting,
\begin{equation}\label{7.40}
\widetilde{v}^{i}=\frac{\widetilde{u}\widetilde{f}^{i}}{\sqrt{1+\widetilde{u}^{2}|\widetilde{\nabla} \widetilde{f}|_{g_{2}}^{2}}},\text{ }\text{ }\text{ }\text{ }\text{ }\text{ }\text{ }
\widetilde{w}^{i}=\frac{\widetilde{u}\widetilde{f}^{i} +\widetilde{u}^{-1}\widetilde{Y}^{i}}{\sqrt{1+\widetilde{u}^{2}|\widetilde{\nabla} \widetilde{f}|_{g_{2}}^{2}}},
\end{equation}
and
\begin{equation}\label{7.41}
\widetilde{Q}_{i}=\widetilde{Y}^{j}\hat{\nabla}_{ij} \widetilde{f}
-\widetilde{u}g_{3}^{jl}\widetilde{f}_{l}(k_{3})_{ij}
+\widetilde{w}^{j}(k_{2}-\pi_{2})_{ij}
+\widetilde{u}\widetilde{f}_{i}\widetilde{w}^{l}\widetilde{w}^{j}(k_{2}-\pi_{2})_{lj}\sqrt{1+\widetilde{u}^{2}|\widetilde{\nabla} \widetilde{f}|_{g_{2}}^{2}}.
\end{equation}

The choice of $\widetilde{u}$ will be give below. Here we simply state the
asymptotics (cf. \cite{ChaKhuriSakovich}) which allow for an appropriate solution to equation \eqref{7.37}. Namely
\begin{equation}\label{7.42}
\widetilde{u}= 1 + \frac{\mathcal{C}_{1}}{r} + \frac{\mathcal{C}_{2}(\theta, \phi)}{r^{2}} + \frac{\mathcal{C}_{3}(\theta, \phi)}{r^{3}} + O_{3}(r^{-4})\text{ }\text{ }\text{ as }\text{ }\text{ }r\rightarrow\infty\text{ }\text{ }\text{ in }\text{ }\text{ }(M_2)_{end}^{+},
\end{equation}
where $\mathcal{C}_{1} = -m_{3}$ is the ADM mass of the third deformation.  On the other end we impose the following asymptotic conditions
\begin{equation}\label{7.420}
\widetilde{u}=r^{2}+ \mathcal{C}_{4}(\theta, \phi)r^{3} + o_{3}(r^{\frac{7}{2}})
\text{ }\text{ }\text{ as }\text{ }\text{ }r\rightarrow 0\text{ }\text{ }\text{ in asymptotically hyperboloidal }(M_{2})_{end}^{-},
\end{equation}
\begin{equation}\label{7.421}
\widetilde{u}=r+o_{3}(r^{\frac{3}{2}})\text{ }\text{ }\text{ as }\text{ }\text{ }r\rightarrow 0\text{ }\text{ }\text{ in asymptotically cylindrical }(M_{2})_{end}^{-},
\end{equation}
where $\mathcal{C}_{2}$, $\mathcal{C}_{3}$, and $\mathcal{C}_{4}$ are functions on the sphere $S^{2}$. The desired asymptotics for $\widetilde{f}$ are given by (cf. \cite{ChaKhuriSakovich})
\begin{equation}\label{7.43}
\widetilde{f}(r, \theta, \phi)=\sqrt{1+r^{2}} + \mathcal{A}\log{r} + \mathcal{B}(\theta, \phi)+\frac{\mathcal{D}_{1}(\theta,\phi)}{r} + \frac{\mathcal{D}_{2}(\theta, \phi)}{r^{2}} + \hat{f}(r, \theta, \phi)
\end{equation}
as $r\rightarrow\infty$ in $(M_{2})_{end}^{+}$,
where $\mathcal{A}=2m_{2}$ and $\mathcal{B}$, $\mathcal{D}_{1}$, $\mathcal{D}_{2}$ are functions on $S^{2}$ satisfying
\begin{align} \label{7.44}
\begin{split}
\Delta_{\sigma}\mathcal{B} = &\frac{1}{2}\left[Tr_{\sigma}( \mathbf{m}_{2}^{g} + 2\mathbf{m}_{2}^{k})+2\mathbf{m}_{2}^{r}\right]
+ \frac{1}{8\pi}\int_{S^{2}}\left[Tr_{\sigma}( \mathbf{m}_{2}^{g} + 2\mathbf{m}_{2}^{k})+2\mathbf{m}_{2}^{r}\right], \\
\mathcal{D}_{1} (\theta, \phi) =& \mathcal{A} \mathcal{C}_{1} - \mathcal{C}_{1}^{2} + 2 \mathcal{C}_{2},  \\ %= -3\overline{m}^{2} + 2C_{2}%
6\mathcal{D}_{2} (\theta, \phi) =& \mathbf{m}^{r}-2\mathbf{m}_{2}^{r}
+\frac{1}{8\pi}\int_{S^{2}}\left[Tr_{\sigma}( \mathbf{m}_{2}^{g} + 2\mathbf{m}_{2}^{k})+2\mathbf{m}_{2}^{r}\right] \\
&+ 6\mathcal{C}_{3}-2\mathcal{C}_{1}\mathcal{D}_{1} + 4\mathcal{A} \mathcal{C}_{2} -2\mathcal{A} \mathcal{C}_{1}^{2} -6\mathcal{C}_{1}\mathcal{C}_{2} +2\mathcal{C}_{1}^{3},
%=\frac{1}{6} \left( \textbf{b} + \frac{3Tr_{\sigma} \textbf{m}}{2} + 6C_{3} +18 \overline{m} C_{2} -12 \overline{m}^{3} \right)%
\end{split}
\end{align}
with
\begin{equation} \label{7.45}
\hat{f}=O_{4}(r^{-3}).
\end{equation}
On the other end
\begin{equation}\label{7.46}
\widetilde{f} = - \frac{1}{3r^{3}} + \frac{\mathcal{C}_{4}(\theta, \phi)}{r^{2}} + o_{4}(r^{-\frac{3}{2}})\text{ }\text{ }\text{ as }\text{ }\text{ }r\rightarrow 0\text{ }\text{ }\text{ in asymptotically hyperboloidal }(M_{2})_{end}^{-},
\end{equation}
\begin{equation}\label{7.47}
\sum_{p=1}^{4}
|\widetilde{\nabla}^{p} \widetilde{f}|_{g_{2}}=o(r^{\frac{1}{2}})\text{ }\text{ }\text{ as }\text{ }\text{ }r\rightarrow 0\text{ }\text{ }\text{ in asymptotically cylindrical }(M_{2})_{end}^{-}.
\end{equation}

\begin{lemma}\label{lemma7.3}
The initial data set $(M_3,g_3,k_3)$ is axially symmetric with two ends, one designated asymptotically flat $(M_{3})_{end}^{+}$ and the other $(M_{3})_{end}^{-}$ either asymptotically flat or asymptotically cylindrical depending on whether $(M_{2})_{end}^{-}$ is asymptotically hyperboloidal or asymptotically cylindrical. Furthermore the mass and angular momentum are preserved, that is $m_{3}=m_{2}$ and $\mathcal{J}_{3}=\mathcal{J}_{2}$.
\end{lemma}

\begin{proof}
It is clear that the new data are axisymmetric.  Moreover, according to \cite{ChaKhuriSakovich} $(M_{3})_{end}^{+}$ is asymptotically flat with $m_{3}=2m_{2}+\mathcal{C}_{1}$ and $\mathcal{J}_{3}=\mathcal{J}_{2}$. Since $\mathcal{C}_{1}=-m_{3}$ we have $m_3=m_2$.

It remains to establish the asymptotics of $(M_{3})_{end}^{-}$. When
$(M_2)_{end}^{-}$ is asymptotically hyperboloidal, as $r\rightarrow 0$ it holds that
\begin{equation} \label{7.472}
\begin{split}
(g_{3})_{rr}
&= (g_{2})_{rr} + 2 \widetilde{f}_{r}(g_{2})_{r \phi}\widetilde{Y}^{\phi} + \left(\widetilde{u}^{2} + (\widetilde{Y}^{\phi})^{2}(g_{2})_{\phi \phi} \right)\widetilde{f}_{r}^{2} \\
&= r^{-4} + o_{2}(r^{-\frac{7}{2}}),
\end{split}
\end{equation}
\begin{equation} \label{7.473}
\begin{split}
(g_{3})_{r \alpha}
&= (g_{2})_{r \alpha} + \widetilde{f}_{r}(g_{2})_{\alpha \phi}\widetilde{Y}^{\phi} + \widetilde{f}_{\alpha}(g_{2})_{r \phi}\widetilde{Y}^{\phi} + \left(\widetilde{u}^{2} + (\widetilde{Y}^{\phi})^{2}(g_{2})_{\phi \phi} \right)\widetilde{f}_{r}\widetilde{f}_{\alpha} \\
&= o_{2}(r^{-\frac{5}{2}}),
\end{split}
\end{equation}
\begin{equation} \label{7.474}
\begin{split}
(g_{3})_{\alpha \beta}
&= (g_{2})_{\alpha \beta} + \widetilde{f}_{\alpha}(g_{2})_{\beta \phi}\widetilde{Y}^{\phi} + \widetilde{f}_{\beta}(g_{2})_{\alpha \phi}\widetilde{Y}^{\phi} + \left(\widetilde{u}^{2} + (\widetilde{Y}^{\phi})^{2}(g_{2})_{\phi \phi} \right)\widetilde{f}_{\alpha}\widetilde{f}_{\beta} \\
&= r^{-2} \sigma_{\alpha \beta} + o_{2}(r^{-\frac{3}{2}}).
\end{split}
\end{equation}
Moreover, \eqref{7.34} together with \eqref{7.361} and \eqref{7.420} yield
\begin{equation} \label{7.475}
(k_{3})_{rr}= O_{1}(r^{-\frac{1}{2}}), \text{ }\text{ }\text{ }\text{ }\text{ }\text{ }\text{ }\text{ }\text{ }\text{ }
(k_{3})_{r\alpha} = O_{1}(1),\text{ }\text{ }\text{ }\text{ }\text{ }\text{ }\text{ }\text{ }\text{ }\text{ }
(k_{3})_{\alpha \beta} = O_{1}(r).
\end{equation}
It is now straightforward to see that \eqref{7.472}-\eqref{7.475} satisfy \eqref{7.1} after applying Kelvin inversion, showing that $(M_{3})^{-}_{end}$ is asymptotically flat.
Similar considerations may be used in the asymptotically cylindrical case.
\end{proof}

Combining \eqref{7.17}, \eqref{7.32}, and \eqref{7.38} leads to
\begin{equation} \label{7.48}
\begin{split}
2\mu_{3} &= R_{3}-|k_{3}|_{g_{3}}^{2} \\
&= 2(\mu_{2}-J_{2}(\widetilde{v}))
+|k_{2}- \pi_{2}|_{g_{2}}^{2}
+2\widetilde{u}^{-1}\operatorname{div}_{g_{3}}(\widetilde{u}\widetilde{Q}) \\
&= 2(\mu-J(v)) - 2J_{1}(\widetilde{v})
+|k-\pi|_{g}^{2}+|k_{2}- \pi_{2}|_{g_{2}}^{2}
+2u^{-1}\operatorname{div}_{g_{1}}(uQ)
+2\widetilde{u}^{-1}\operatorname{div}_{g_{3}}(\widetilde{u}\widetilde{Q}).
\end{split}
\end{equation}
Furthermore a computation in Appendix A shows that
\begin{equation} \label{7.49}
J_{1}(\widetilde{v})
= \left(\operatorname{div}_{g_{1}} k_{1}\right)(\widetilde{v})
= u^{-1} \operatorname{div}_{g_{1}} \left(u k_{1}(\widetilde{v}, \cdot)\right)
-g_{1} \left( k_{1}, \pi_{2} \right)
+ \frac{ \operatorname{div}_{g_{1}} ( k_{1}( \widetilde{Y}^{\phi} \partial_{\phi}, \cdot)  )}{\widetilde{u}\sqrt{1+ \widetilde{u}^{2}|\widetilde{\nabla} \widetilde{f}|_{g_{2}}^{2}}}.
\end{equation}
This, together with $k_2=k_1+g_1$ yields
\begin{equation} \label{7.50}
\begin{split}
R_{3}-|k_{3}|_{g_{3}}^{2}
=&2(\mu-J(v)) +|k-\pi|_{g}^{2}
+ \left(|k_{2}|_{g_{2}}^{2} + |\pi_{2}|_{g_{2}}^{2} - 2 Tr_{g_{2}} \pi_{2} \right) \\
&+2u^{-1}\operatorname{div}_{g_{1}}(uQ-uk_{1}(\widetilde{v}, \cdot))
- \frac{ 2\operatorname{div}_{g_{1}} ( k_{1}( \widetilde{Y}^{\phi} \partial_{\phi}, \cdot) )}{\widetilde{u}\sqrt{1+ \widetilde{u}^{2}|\widetilde{\nabla} \widetilde{f}|_{g_{2}}^{2}}}
+2\widetilde{u}^{-1}\operatorname{div}_{g_{3}}(\widetilde{u}\widetilde{Q})\\
\geq &2(\mu-J(v)) +|k-\pi|_{g}^{2}
+2u^{-1}\operatorname{div}_{g_{1}}(uQ-uk_{1}(\widetilde{v}, \cdot))\\
&- \frac{ 2\operatorname{div}_{g_{1}} ( k_{1}( \widetilde{Y}^{\phi} \partial_{\phi}, \cdot) )}{\widetilde{u}\sqrt{1+ \widetilde{u}^{2}|\widetilde{\nabla} \widetilde{f}|_{g_{2}}^{2}}}
+2\widetilde{u}^{-1}\operatorname{div}_{g_{3}}(\widetilde{u}\widetilde{Q}),
\end{split}
\end{equation}
since maximality of the first deformation and \eqref{7.37} imply that
\begin{equation} \label{7.51}
\begin{split}
|k_{2}|_{g_{2}}^{2} + |\pi_{2}|_{g_{2}}^{2} - 2 Tr_{g_{2}} \pi_{2}
&= |k_{2}|_{g_{2}}^{2} + |\pi_{2}|_{g_{2}}^{2} - 6 \\
&= \left( |k_{2}|_{g_{2}}^{2}-\frac{(Tr_{g_{2}}k_{2})^{2}}{3}\right)
+ \left( |\pi_{2}|_{g_{2}}^{2}-\frac{(Tr_{g_{2}}\pi_{2})^{2}}{3}\right)\geq 0.
\end{split}
\end{equation}

We are now in a position to choose $u$ and $\widetilde{u}$. Lemma \ref{lemma7.3} and simple connectivity ensure that the hypotheses of \cite{Chrusciel,KhuriSokolowsky} are satisfied. This shows that $M_3$ is diffeomorphic to $\mathbb{R}^{3}\setminus\{0\}$ (the same is true for $M$, $M_1$, and $M_2$), and that there exists a global Brill (cylindrical) coordinate system $(\rho,z,\phi)$ in which the metric takes the simple form
\begin{equation}\label{7.52}
g_{3}=e^{-2\widetilde{U}+2\widetilde{\alpha}}
(d\rho^{2}+dz^{2})
+e^{-2\widetilde{U}}\rho^{2}
(d\phi+A_{\rho}d\rho
+A_{z}dz)^{2}.
\end{equation}
The arguments in \cite{ChaKhuri1} combined with \eqref{7.50} then produce
\begin{equation}\label{7.53}
\begin{split}
m_{3}-\mathcal{M}(\widetilde{U},\widetilde{\omega})
&\geq
 \frac{1}{8\pi}\int_{M_{2}}\frac{e^{\widetilde{U}}\sqrt{1 + \widetilde{u}^{2}|\widetilde{\nabla} \widetilde{f}|_{g_{2}}^{2}} }{u}\operatorname{div}_{g_{2}}(uQ-uk_{1}(\widetilde{v}, \cdot))
- \frac{e^{\widetilde{U}}}{\widetilde{u}} \operatorname{div}_{g_{2}} ( k_{1}( \widetilde{Y}^{\phi} \partial_{\phi}, \cdot)  ) \\
&+\frac{1}{8\pi}\int_{M_{3}}\frac{e^{\widetilde{U}}}{\widetilde{u}}
\operatorname{div}_{g_{3}}(\widetilde{u} \widetilde{Q}),
\end{split}
\end{equation}
where the twist potential $\widetilde{\omega}$ encodes angular momentum contributions
and the mass functional is given by
\begin{equation}\label{7.54}
\mathcal{M}(\widetilde{U},\widetilde{\omega})
=\frac{1}{32\pi}\int_{\mathbb{R}^{3}}4|\partial \widetilde{U}|^{2}+\frac{e^{4\widetilde{U}}}{\rho^{4}}
|\partial\widetilde{\omega}|^{2}.
\end{equation}
This then motivates the choices
\begin{equation} \label{7.55}
\widetilde{u} = e^{\widetilde{U}},
\end{equation}
and
\begin{equation} \label{7.56}
u=\widetilde{u} \sqrt{1 + \widetilde{u}^{2}|\widetilde{\nabla} \widetilde{f}|_{g_{2}}^{2}}.
\end{equation}

\begin{theorem}\label{thm7.4}
Let $(M,g,k)$ be a smooth, simply connected, axially symmetric initial data set satisfying the dominant energy condition $\mu\geq|J|_{g}$ and $J(\eta)=0$, having two ends, one designated asymptotically AdS hyperbolic and the other either asymptotically AdS hyperbolic or asymptotically cylindrical. If the system of equations \eqref{7.9}, \eqref{7.13}, \eqref{7.36}, \eqref{7.37}, \eqref{7.55} and \eqref{7.56} admits a smooth solution $(u,\widetilde{u},Y^{\phi},\widetilde{Y}^{\phi},f,\widetilde{f})$ satisfying the asymptotics described above, then
\begin{equation}\label{7.57}
m \geq\sqrt{|\mathcal{J}|}.
\end{equation}
\end{theorem}

\begin{proof}
The right-hand side of \eqref{7.53} results in boundary terms which are shown
to vanish in Appendix A, therefore $m_{3} \geq \mathcal{M}(\widetilde{U}, \widetilde{w})$. Furthermore it is proved in \cite{Dain0, SchoenZhou} that $\mathcal{M}(\widetilde{U}, \widetilde{w}) \geq \sqrt{|\mathcal{J}_{3}|}$.
Now the desired result follows from Lemma \ref{lemma7.1}, Lemma \ref{lemma7.2}, and Lemma \ref{lemma7.3} which yield $m_{3}=m$ and $\mathcal{J}_{3}= \mathcal{J}$.
\end{proof}

\begin{remark}\label{remark0}
As noted in the introduction to this section, the case of equality for \eqref{7.57} is
not expected to occur. Further evidence in this direction is supported by the following arguments. Suppose that equality holds, that is $m=\sqrt{|\mathcal{J}|}$. This implies equality in \eqref{7.51}, which together with \eqref{7.31} leads to $k_{1} \equiv 0$.
We should then have $Y^{\phi} \equiv 0$, and hence $\mathcal{J}=m=0$. However this is a contradiction to the positive mass theorem since $M\cong\mathbb{R}^{3}\setminus\{0\}$ is not topologically trivial.
\end{remark}

In order to further support the above procedure we show that solutions to equation \eqref{7.13} exist with the desired asymptotics in $M_{end}^{+}$. Note that the existence of appropriate solutions to \eqref{7.36} and \eqref{7.37} has already been studied in \cite{ChaKhuri1,ChaKhuriSakovich}, and similar techniques can be applied to \eqref{7.9}.

\begin{theorem}\label{thm7.5}
Given a smooth positive function $u$ satisfying \eqref{7.21}, \eqref{7.22}, \eqref{7.23} there exists a smooth solution $f$ of \eqref{7.13} satisfying \eqref{7.14}.
\end{theorem}

\begin{proof}
This is based on Theorem 6.3 in \cite{ChaKhuriSakovich}, and thus we only outline
the main arguments which require modification. Equation \eqref{7.13} may be rewritten as
\begin{equation}\label{7.58}
 \Delta_{g} f  + 2 \left \langle \frac{\nabla u}{u}, \nabla f \right\rangle
 - (Tr_g k)\left(\sqrt{u^{-2} + |\nabla f|^2_g} - u^{-1}\right) -u^{-1}(Tr_{g}k) = 0.
\end{equation}
Let $r$ be a smooth positive function on $M$ which agrees with the asymptotic radial
coordinate in both ends, and satisfies $|\nabla r|\neq 0$. Consider a `radial function' $h=h(r)$ and compute
\begin{equation} \label{7.59}
 \begin{split}
  &\Delta_{g} h  + 2 \left \langle \frac{\nabla u}{u}, \nabla h \right \rangle\\
%  = &g^{rr}h_{0}'' - g^{rr} \left( \Gamma^{r}_{rr} + 2(g^{rr})^{-1}g^{r \alpha}\Gamma^{r}_{r \alpha} + (g^{rr})^{-1} g^{\alpha \beta}\Gamma^{r}_{\alpha \beta} -  \frac{2\partial_{r} u}{u} -\frac{2g^{r\alpha}\partial_{\alpha}u}{ug^{rr}}\right) h_{0}' \\
=& \frac{g^{rr}}{\sqrt{1+r^2}} \left[ \zeta' -
\left( \frac{r}{1+r^{2}}+ \Gamma^{r}_{rr} + 2(g^{rr})^{-1}g^{r \alpha}\Gamma^{r}_{r \alpha} + (g^{rr})^{-1} g^{\alpha \beta}\Gamma^{r}_{\alpha \beta} -  \frac{2\partial_{r} u}{u} -\frac{2g^{r\alpha}\partial_{\alpha}u}{ug^{rr}}\right)\zeta \right],
 \end{split}
\end{equation}
where $\zeta = \sqrt{1 + r^{2}} h'$.
Furthermore let $\varTheta(r)$, $\Phi(r)>0$, $\varUpsilon(r)>0$ be bounded radial functions such that
\begin{equation} \label{7.60}
 \frac{\sqrt{1+r^{2}}\left| u^{-1}(Tr_g k)\right|}{g^{rr}}\leq \varUpsilon,
\end{equation}
\begin{equation} \label{7.61}
\left|\left( \frac{r}{1+r^{2}}+ \Gamma^{r}_{rr} + 2(g^{rr})^{-1}g^{r \alpha}\Gamma^{r}_{r \alpha} + (g^{rr})^{-1} g^{\alpha \beta}\Gamma^{r}_{\alpha \beta} -  \frac{2\partial_{r} u}{u} -\frac{2g^{r\alpha}\partial_{\alpha}u}{ug^{rr}}\right)+ \frac{|Tr_g k|}{\sqrt{g^{rr}}} + \varTheta \right| \leq \Phi,
\end{equation}
and
\begin{equation} \label{7.62}
\varTheta (r)= \frac{4}{r}-\frac{2u_{0}}{r^2} + O(r^{-3}),\text{ }\text{ }\text{ }\text{ }\text{ }\text{ }\Phi(r) = O(r^{-3}), \text{ }\text{ }\text{ }\text{ }\text{ }\text{ }\varUpsilon(r)=O(r^{-5})\text{ }\text{ }\text{ }\text{ in }\text{ }\text{ }\text{ }M_{end}^{+}.
\end{equation}

Define
\begin{equation} \label{7.63}
 f_+(r) =  - \int_r ^\infty \frac{\zeta_+(s)}{\sqrt{1+s^2}} ds,
\end{equation}
where
\begin{equation} \label{7.64}
\zeta_+(r)= - e^{-\int_{0}^{r} (\varTheta (s) - \Phi(s))ds} \int_{0}^{r} \varUpsilon(s) e^{\int_{0}^s (\varTheta (t) - \Phi(t))dt} ds \leq 0
\end{equation}
is a solution of the ordinary differential equation
\begin{equation}\label{7.65}
\zeta_+' + (\varTheta - \Phi) \zeta_+ + \varUpsilon = 0.
\end{equation}
It is then straightforward to check that $f_+>0$ is a supersolution of \eqref{7.58}, and has the desired asymptotics. Similarly it may be shown that $f_{-}=-f_{+}$ is a subsolution. Existence of a solution satisfying \eqref{7.14} may now be obtained via
an exhaustion argument as in \cite{ChaKhuri1}, using derivative estimates established with the techniques of \cite{SakovichThesis,Sakovich}.
\end{proof}

\section{The Mass-Angular Momentum-Charge Inequality}
\label{sec8} \setcounter{equation}{0}
\setcounter{section}{8}

In this section we propose a new inequality between mass, angular momentum, and charge in the asymptotically AdS hyperbolic setting, namely
\begin{equation}\label{8.0}
m^2\geq\frac{\mathcal{Q}^2+\sqrt{\mathcal{Q}^4+4\mathcal{J}^2}}{2}.
\end{equation}
This inequality in certain regimes is stronger/weaker than the BPS bound
\begin{equation}\label{8.00}
m\geq m_{bps}:=\sqrt{\lambda}|\mathcal{J}|+|\mathcal{Q}|
\end{equation}
established in \cite{KosteleckyPerry}. We conjecture that the optimal mass-angular momentum-charge inequality in this setting is given by
\begin{align}\label{8.000}
\begin{split}
m\geq m_{ext}& := \frac{1}{3\sqrt{6\lambda}}\left[\sqrt{\left(1+\frac{\lambda\mathcal{J}^2}{m^2}\right)^2
+12\lambda\left(\frac{\mathcal{J}^2}{m^2}+\left(1-\frac{\lambda\mathcal{J}^2}{m^2}\right)
\mathcal{Q}^2\right)}
+2\left(1+\frac{\lambda\mathcal{J}^2}{m^2}\right)\right]\\
& \cdot\left[\sqrt{\left(1+\frac{\lambda\mathcal{J}^2}{m^2}\right)^2
+12\lambda\left(\frac{\mathcal{J}^2}{m^2}+\left(1-\frac{\lambda\mathcal{J}^2}{m^2}\right)
\mathcal{Q}^2\right)}
-\left(1+\frac{\lambda\mathcal{J}^2}{m^2}\right)\right]^{1/2},
\end{split}
\end{align}
which is saturated for extreme Kerr-Newman-AdS black holes \cite{CaldarelliKlemm}.
This inequality is motivated by the generalized Penrose inequality
\begin{align}\label{smarr}
\begin{split}
m^2\geq &\frac{A}{16\pi}+\frac{\pi}{A}\left(4\mathcal{J}^2+\mathcal{Q}^4\right)
+\frac{\mathcal{Q}^2}{2}
+\lambda\left[\mathcal{J}^2 +\frac{A}{8\pi}\left(\mathcal{Q}^2+\frac{A}{4\pi}
+\frac{\lambda A^2}{32\pi^2}\right)\right]\\
:= & F_{1}(A,\mathcal{J},\mathcal{Q})
+\lambda F_{2}(A,\mathcal{J},\mathcal{Q},\lambda)
\end{split}
\end{align}
arising from the Smarr formula \cite{CaldarelliCognolaKlemm},
since $m_{ext}^2$ is the minimum over area $A$ (for fixed angular momentum and charge) of the right-hand side.
Here $\lambda = -\frac{\Lambda}{3}$, $F_{2}$ is the expression within the brackets and $F_{1}$ is the expression without $\lambda$; equality is achieved for the Kerr-Newman-AdS black hole.  Let $A_0$ be the area value at which $F_1$ achieves its minimum for fixed $\mathcal{J}$ and $\mathcal{Q}$, then
\begin{equation}\label{comparison}
m_{ext}^2=\inf_{A} [F_1(A,\mathcal{J},\mathcal{Q})
+\lambda F_2(A,\mathcal{J},\mathcal{Q},\lambda)]
>  F_{1}(A_0,\mathcal{J},\mathcal{Q})
=\frac{\mathcal{Q}^2+\sqrt{\mathcal{Q}^4+4\mathcal{J}^2}}{2}.
\end{equation}
It follows that \eqref{8.000} implies \eqref{8.0}. Furthermore, as explained in \cite{HristovToldaVandoren}
\begin{equation}\label{comparison1}
m_{ext}>m_{bps}\quad\text{ if }\quad |\mathcal{Q}|\neq |\mathcal{Q}_0|,\quad\quad\quad
m_{ext}=m_{bps}\quad\text{ if }\quad |\mathcal{Q}|=|\mathcal{Q}_0|,
\end{equation}
where
\begin{equation}
\mathcal{Q}_0=\left(\frac{|\mathcal{J}|}{\sqrt{\lambda}m}\right)^{\frac{1}{2}}
\frac{1}{1-\frac{\sqrt{\lambda}|\mathcal{J}|}{m}}.
\end{equation}
In particular, it always holds that $m_{ext}\geq m_{bps}$ and therefore \eqref{8.000} implies the BPS bound \eqref{8.00} as well. We note also that
it is not expected that any proper black hole solution saturates the bounds \eqref{8.0} or \eqref{8.00}, except in the case when $|\mathcal{Q}|=|\mathcal{Q}_0|$, due to the strict inequalities in \eqref{comparison} and \eqref{comparison1}.

Let $(M, g, k, E, B)$ be a simply connected, axisymmetric initial data set for the Einstein-Maxwell equations with two ends, one designated asymptotically AdS hyperbolic $M_{end}^{+}$ and the other $M_{end}^{-}$ either asymptotically AdS hyperbolic or asymptotically cylindrical. The electric and magnetic fields $E$ and $B$ are assumed to be divergence free and axisymmetric
\begin{equation} \label{8.1}
\operatorname{div}_{g}E=\operatorname{div}_{g}B=0,\quad\quad\quad\quad\mathfrak{L}_{\eta} E=\mathfrak{L}_{\eta} B=0.
\end{equation}
Asymptotics for the electromagnetic field are given by \eqref{6} in $M_{end}^{+}$, and in $M_{end}^{-}$ they are given respectively in the asymptotically AdS hyperbolic and asymptotically cylindrical cases by
\begin{equation}\label{8.3}
E_r=O_{1}(r), \quad\quad\quad E_{\alpha}=O_{1}(r),
\end{equation}
and
\begin{equation}\label{8.4}
E_r=O_{1}(r^{-1}), \quad\quad\quad E_{\alpha}=O_{1}(1),
\end{equation}
where Greek indices denote coordinates on the sphere $S^2$.
Angular momentum is defined by
\begin{equation}\label{8.5}
\mathcal{J}=\frac{1}{8\pi}\int_{S}(k_{ij}-(Tr_{g} k)g_{ij})\nu^{i}_{g}\eta^{j}-\frac{1}{4\pi}\int_{S}\psi_{B} E_{i}\nu_{g}^{i},
\end{equation}
where $S$ is any surface homologous to the sphere at infinity, with unit outer normal $\nu_{g}$, and $\psi_{B}$ is the potential for the magnetic field (see \cite{DainKhuriWeinsteinYamada} and \cite{KhuriWeinstein1}).
In analogy with \eqref{7.4}, this definition will be well-defined if the
following condition is satisfied
\begin{equation}
J_{EM}^{i}\eta_{i}=0.
\end{equation}

We will describe a three step deformation procedure similar to that given
in Section \ref{sec7}. In particular this transformation is determined by
six functions $(u,\widetilde{u}, Y^{\phi}, \widetilde{Y}^{\phi}, f, \widetilde{f})$ which,
except for $Y^{\phi}$ and $\widetilde{Y}^{\phi}$, satisfy the same equations as in the previous
section, namely \eqref{7.13}, \eqref{7.37}, \eqref{7.55} and \eqref{7.56}. The equations for $Y^{\phi}$ and $\widetilde{Y}^{\phi}$ will be embellishments of \eqref{7.9} and \eqref{7.36} which account
for contributions coming from the electromagnetic field. Thus, the only significant new elements of the
process involve deformations of the Maxwell field which we now describe; unless otherwise noted,
the notation and definitions of Section \ref{sec7} remain valid here.

Let $(e_{1}, e_{2}, e_{3}=|\eta|^{-1}\eta)$ be an orthonormal frame for $(M,g)$, and consider an
the asymptotically AdS hyperbolic initial data set $(M_{1}, g_{1}, k_{1}, E_{1}, B_{1})$ where
\begin{equation}\label{8.6}
\begin{split}
&(E_{1})(e_{i})=\frac{E(e_{i})}{\sqrt{\volg}},\text{ }\text{ }\text{ }\text{ }\text{ }
(B_{1})(e_{i})=\frac{B(e_{i})}{\sqrt{\volg}}\text{ }\text{ }\text{ }\text{ for }\text{ }\text{ }\text{ }i=1,2, \\
&(E_{1})(e_{3})=(B_{1})(e_{3})=0.
\end{split}
\end{equation}
Then according to \cite{ChaKhuri2} it holds that
\begin{equation}\label{8.7}
\operatorname{div}_{g_{1}}E_{1}
= \frac{\operatorname{div}_{g} E}{\sqrt{\volg}}=0, \text{ }\text{ }\text{ }\text{ }\text{ }\text{ }\text{ }\text{ }
\operatorname{div}_{g_{1}}B_{1} = \frac{\operatorname{div}_{g} B}{\sqrt{\volg}}=0.
\end{equation}
In addition, we may then set the equation for $Y^{\phi}$ to be given by
\begin{equation}\label{8.8}
J^{1}_{EM}(\eta) = \operatorname{div}_{g_{1}} k_{1}(\eta) + 2 E_{1} \times B_{1} (\eta)=0,
\end{equation}
which ensures a well-defined angular momentum. A computation also yields the energy density of the non-electromagnetic matter fields
\begin{equation}\label{8.10}
\begin{split}
2 \mu^{1}_{EM}
=& R_{1}-|k_{1}|_{g_{1}}^{2}-2(|E_{1}|_{g_{1}}^{2}+|B_{1}|_{g_{1}}^{2}) - 2 \Lambda \\
=&2(\mu_{EM}-J_{EM}(v))+|k-\pi|_{g}^{2}+2u^{-1}\operatorname{div}_{\overline{g}}(uQ)\\
&+2\left(E(e_{3})- v\times B(e_{3})\right)^{2}
+ 2\left(B(e_{3})+ v\times E(e_{3})\right)^{2}.
\end{split}
\end{equation}
With appropriate asymptotics, the arguments of Lemma \ref{lemma7.1} and \cite{ChaKhuri2} show that the mass, angular momentum, and charge are conserved $m_{1}=m$, $\mathcal{J}_{1}=\mathcal{J}$, $\mathcal{Q}^{1}_{e}=\mathcal{Q}_{e}$, $\mathcal{Q}^{1}_{b}=\mathcal{Q}_{b}$.

The second step in the deformation process yields an asymptotically hyperboloidal initial data set $(M_{2}, g_{2}, k_{2}, E_{2}, B_{2})$ such that
\begin{equation} \label{8.11}
E_{2} \equiv E_{1}, \text{ }\text{ }\text{ }\text{ }\text{ }\text{ }
B_{2} \equiv B_{1}.
\end{equation}
Clearly $\mu^{2}_{EM}=\mu^{1}_{EM}$, and $J^{2}_{EM}=J^{1}_{EM}$, and in particular
\begin{equation} \label{8.12}
J^{2}_{EM} (\eta) =0.
\end{equation}
Mass, angular momentum, and charge are again conserved $m_{2}=m_{1}$, $\mathcal{J}_{2}=\mathcal{J}_{1}$, $\mathcal{Q}^{2}_{e}=\mathcal{Q}^{1}_{e}$, $\mathcal{Q}^{2}_{b}=\mathcal{Q}^{1}_{b}$.

The last step in the deformation process produces an asymptotically flat initial data set $(M_{3}, g_{3}, k_{3}, E_{3}, B_{3})$ in which
\begin{equation}\label{8.13}
\begin{split}
&(E_{3})(\widetilde{e}_{i})=\frac{E_{2}(\widetilde{e}_{i})}{\sqrt{1+\widetilde{u}^{2}|\widetilde{\nabla} \widetilde{f}|_{g_{2}}^{2}}},\text{ }\text{ }\text{ }\text{ }\text{ }
(B_{3})(\widetilde{e}_{i})=\frac{B_{2}(\widetilde{e}_{i})}{\sqrt{1+\widetilde{u}^{2}|\widetilde{\nabla} \widetilde{f}|_{g_{2}}^{2}}},\text{ }\text{ }\text{ }\text{ }\text{ }i=1,2, \\
&(E_{3})(\widetilde{e}_{3})=(B_{3})(\widetilde{e}_{3})=0,
\end{split}
\end{equation}
where $(\widetilde{e}_{1}, \widetilde{e}_{2}, \widetilde{e}_{3}=|\eta|^{-1}\eta)$ is an orthonormal frame for $(M_{2},g_{2})$. As above, this implies that
\begin{equation}\label{8.14}
\operatorname{div}_{g_{3}}E_{3}
= \frac{\operatorname{div}_{g_{2}} E_{2}}{\sqrt{1+\widetilde{u}^{2}|\widetilde{\nabla} \widetilde{f}|_{g_{2}}^{2}}}=0, \text{ }\text{ }\text{ }\text{ }\text{ }\text{ }\text{ }\text{ }
\operatorname{div}_{g_{3}}B_{3} = \frac{\operatorname{div}_{g_{2}} B_{2}}{\sqrt{1+\widetilde{u}^{2}|\widetilde{\nabla} \widetilde{f}|_{g_{2}}^{2}}}=0.
\end{equation}
The equation for $\widetilde{Y}^{\phi}$ is then chosen to be
\begin{equation}\label{8.16}
J^{3}_{EM}(\eta) = \operatorname{div}_{g_{3}} k_{3}(\eta) + 2 E_{3} \times B_{3} (\eta)=0,
\end{equation}
which guarantees the existence of a charged twist potential and well-defined angular momentum.
The scalar curvature formula \eqref{7.50}, and a computation in \cite{ChaKhuri2} together with \eqref{8.10} yield the energy density of the non-electromagnetic matter fields
\begin{equation}\label{8.17}
\begin{split}
2 \mu^{3}_{EM}
=& R_{3} -|k_{3}|_{g_{3}}^{2}-2 (|E_{3}|^{2}_{g_{3}}+|B_{3}|^{2}_{g_{3}} )  \\
=&2(\mu_{EM}-J_{EM}(v)) +|k-\pi|_{g}^{2}
+ \left(|k_{2}|_{g_{2}}^{2} + |\pi_{2}|_{g_{2}}^{2} - 2 Tr_{g_{2}} \pi_{2} \right) \\
&+2u^{-1}\operatorname{div}_{g_{1}}(uQ-uk_{1}(\widetilde{v}, \cdot))
- 2 \left(\widetilde{u}\sqrt{1+\widetilde{u}^{2}|\widetilde{\nabla} \widetilde{f}|_{g_{2}}^{2}}\right)^{-1}\operatorname{div}_{g_{1}} ( k_{1}( \widetilde{Y}^{\phi} \partial_{\phi}, \cdot)  )
\\
&+2\widetilde{u}^{-1}\operatorname{div}_{g_{3}}(\widetilde{u}\widetilde{Q}) +2\left(E(e_{3})- v\times B(e_{3})\right)^{2}
+ 2\left(B(e_{3})+ v\times E(e_{3})\right)^{2} \\
&+2\left(E_{2}(\widetilde{e}_{3})- \widetilde{v}\times B_{2}(\widetilde{e}_{3})\right)^{2}
+ 2\left(B(\widetilde{e}_{3})+ \widetilde{v}\times E(\widetilde{e}_{3})\right)^{2}.
\end{split}
\end{equation}
With appropriate asymptotics, the arguments of Lemma \ref{lemma7.3} and \cite{ChaKhuri2} show that the mass, angular momentum, and charge are conserved $m_{3}=m_2$, $\mathcal{J}_{3}=\mathcal{J}_2$, $\mathcal{Q}^{3}_{e}=\mathcal{Q}^{2}_{e}$, $\mathcal{Q}^{3}_{b}=\mathcal{Q}^{2}_{b}$.

\begin{theorem}\label{thm8.1}
Let $(M,g,k,E,B)$ be a smooth, simply connected, axially symmetric initial data set satisfying the charged dominant energy condition $\mu_{EM}\geq|J_{EM}|_{g}$ and $J_{EM}(\eta)=0$, having two ends, one designated asymptotically AdS hyperbolic and the other either asymptotically AdS hyperbolic or asymptotically cylindrical. If the system of equations \eqref{7.13}, \eqref{7.37}, \eqref{7.55}, \eqref{7.56}, \eqref{8.8}, and \eqref{8.16} admits a smooth solution $(u,\widetilde{u},Y^{\phi},\widetilde{Y}^{\phi},f,\widetilde{f})$ satisfying the asymptotics described in Section \ref{sec7}, then
\begin{equation}\label{8.18}
m^2 \geq \frac{\mathcal{Q}^{2}+\sqrt{\mathcal{Q}^{4}+4\mathcal{J}^{2}}}{2}.
\end{equation}
\end{theorem}

\begin{proof}
As in the proof of Theorem \ref{thm7.4}, the arguments from \cite{ChaKhuri2} can be directly applied with the help of \eqref{8.17}. The boundary integrals appearing in the process vanish by Appendix A.
\end{proof}

\begin{remark}
As noted in the introduction to this section, the case of equality for \eqref{8.18} is
not expected to occur. Further evidence in this direction arises from similar arguments
as those presented in Remark \ref{remark0}. In particular, if equality occurs then it should
hold that $\mathcal{J}=0$. Thus, if the initial data has nonzero angular momentum a contradiction
is reached. If the initial data has zero angular momentum, the case of equality in \eqref{8.18} reduces to the case of equality in the positive mass theorem with charge. The discussion at the end of Section \ref{sec6} then suggests that $\mathcal{Q}=0$, and hence $m=0$. The topology of $M$ then yields a contradiction with the positive mass theorem.
\end{remark}

\section{Appendix A: Boundary Integrals}
\label{sec10} \setcounter{equation}{0}
\setcounter{section}{9}

In this section we collect four technical lemmas which are needed to complete the proofs of the main results.

\begin{lemma}\label{lemma9.1}
Under the hypotheses and notation of Theorems \ref{thm3.4} and \ref{thm4.2}
\begin{equation}\label{10.0}
 \lim_{r\rightarrow\infty}\int_{\overline{S}_{r}} u g_{2}(q,\nu_{g_{2}}) = 0.
\end{equation}
\end{lemma}

\begin{proof}
We will follow the arguments in Appendix B of \cite{ChaKhuriSakovich}.
Let $S_{r}$ be coordinate spheres and $(e_{r},e_{\theta},e_{\phi})$ be an orthonormal frame in the designated asymptotically AdS end of $(M,g)$. Then a computation shows that
\begin{equation} \label{10.0.1}
\int_{\overline{S}_{r}} u g_{2}(q,\nu_{g_{2}})
= \int_{\overline{S}_{r}} u g_{1}(q,\nu_{g_{1}})
= \int_{S_{r}}  \frac{u(1 + u^{2}|\nabla_{S}f|^{2})}{\sqrt{1+u^{2}|\nabla f|_{g}^{2}}}
\left(q(e_{r}) - \frac{u^{2}e_{r}(f)}{1+ u^{2}|\nabla_{S}f|^{2}} q(\nabla_{S}f)\right),
\end{equation}
where $\nabla_{S}$ denotes covariant differentiation on $S_{r}$.
Observe that
\begin{equation}\label{10.0.2}
\pi_{rr}= O(r^{-4}), \text{ }\text{ }\text{ }\text{ }\text{ }\text{ }\text{ }
\pi_{r \alpha} = O(r^{-3}), \text{ }\text{ }\text{ }\text{ }\text{ }\text{ }\text{ }
\pi_{\alpha \beta} = O(1).
\end{equation}
This together with \eqref{3} yields
\begin{align} \label{10.0.3}
\begin{split}
q(e_{r}) =& q(\sqrt{1+r^{2}}\partial_{r})+O(r^{-7}) \\
=& \frac{u f^{r}(\pi_{rr}-k_{rr})\sqrt{1+r^{2}}+u f^{\alpha}(\pi_{r\alpha}-k_{r\alpha})\sqrt{1+r^{2}}}{\sqrt{1 + u^{2}|\nabla f|_{g}^{2}}}
+O(r^{-7})\\
=&O(r^{-4}),
\end{split}
\end{align}
\begin{equation}
q(\nabla_{S} f)=  \frac{u g^{\alpha \beta}f_{\alpha}f^{j}(\pi_{\beta j}-k_{\beta j})}{\sqrt{1 + u^{2}|\nabla f|_{g}^{2}}}
= O(r^{-9}),
\end{equation}
and
\begin{equation} \label{10.0.5}
\frac{u(1 + u^{2}|\nabla_{S}f|^{2})}{\sqrt{1+u^{2}|\nabla f|_{g}^{2}}}=O(r).
\end{equation}
The desired result now follows.
\end{proof}

\begin{lemma}\label{lemma9.2}
Under the hypotheses and notation of Theorems \ref{thm7.4} and \ref{thm8.1}
\begin{equation}
\int_{M_{2}}\operatorname{div}_{g_{2}}(uQ-uk_{1}(\widetilde{v}, \cdot)-k_{1}(\widetilde{Y}^{\phi}\partial_{\phi}, \cdot)) =0.
\end{equation}
\end{lemma}

\begin{proof}
According to the divergence theorem
\begin{equation} \label{10.1.0}
\begin{split}
&\int_{M_{2}}\operatorname{div}_{g_{2}}(uQ-uk_{1}(\widetilde{v}, \cdot)-k_{1}(\widetilde{Y}^{\phi}\partial_{\phi}, \cdot)) \\
=&  \lim_{r\rightarrow\infty}\int_{\overline{S}_{r}} \left( uQ(\nu_{g_{1}})-uk_{1}(\widetilde{v},\nu_{g_{1}}) -k_{1}(\widetilde{Y}^{\phi}\partial_{\phi}, \nu_{g_{1}})\right) \\
&- \lim_{r\rightarrow 0}\int_{\overline{S}_{r}} \left( uQ(\nu_{g_{1}})-uk_{1}(\widetilde{v},\nu_{g_{1}})
-k_{1}(\widetilde{Y}^{\phi}\partial_{\phi}, \nu_{g_{1}}) \right).
\end{split}
\end{equation}
A computation in Appendix B of \cite{ChaKhuriSakovich} simplifies the first integral as follows
\begin{equation} \label{10.1.1}
\begin{split}
&\lim_{r\rightarrow\infty}\int_{\overline{S}_{r}} \left( uQ(\nu_{g_{1}})-uk_{1}(\widetilde{v},\nu_{g_{1}})
-k_{1}(\widetilde{Y}^{\phi}\partial_{\phi}, \nu_{g_{1}}) \right) \\
=&\int_{\overline{S}_{\infty}} u(\operatorname{Hess}_{g_{1}} f)(Y, \nu_{g_{1}})
-  k_{1}(u^{2} \barna f, \nu_{g_{1}}) -uk_{1}(\widetilde{v},\nu_{g_{1}})
-k_{1}(\widetilde{Y}^{\phi}\partial_{\phi}, \nu_{g_{1}})\\
& +\int_{S_{\infty}} u\sqrt{1+u^{2}|\nabla f|_{g}^{2}}(k-\pi)(w,e_{r}).
\end{split}
\end{equation}
Using \eqref{7.24} and the axisymmetry of $f$ and $\widetilde{f}$ yield $k_{1}(\barna f, \nu_{g_{1}})= k_{1}(\barna \widetilde{f},\nu_{g_{1}})=0$ on $\overline{S}_{r}$. Moreover $u(\operatorname{Hess}_{g_{1}}  f)(Y, \nu_{g_{1}})= O(r^{-5})$ by \eqref{7.10}, \eqref{7.14}, and \eqref{7.21}. The asymptotics of $\widetilde{Y}^{\phi}$ in \eqref{7.360} together with the definition of angular momentum $\mathcal{J}_{1}$ shows that
\begin{equation}
\lim_{r\rightarrow\infty}\int_{\overline{S}_{r}} k_{1}(\widetilde{Y}^{\phi}\partial_{\phi}, \nu_{g_{1}}) =0.
\end{equation}
Therefore the first integral on the right-hand side of \eqref{10.1.1} vanishes.
Consider now the second integral on the right-hand side of \eqref{10.1.1}.
Note that the asymptotics of $\pi$ are the same as \eqref{10.0.2} since $Y^{\phi}$ falls off sufficiently fast. Furthermore \eqref{7.10}, \eqref{7.14}, and \eqref{7.21} imply that
\begin{equation} \label{10.1.3}
w^{r} = O(r^{-1}), \text{ }\text{ }\text{ }\text{ }\text{ }\\
w^{\theta} = O(r^{-4}), \text{ }\text{ }\text{ }\text{ }\text{ } \\
w^{\phi}  = O(r^{-4}),
\end{equation}
and it is shown in \cite{ChaKhuriSakovich} that
\begin{equation}\label{10.1.4}
e_{r}= \left(\sqrt{1+r^{2}} +O(r^{-2}) \right) \partial_{r} + O(r^{-4}) \partial_{\theta}
+ O(r^{-4}) \partial_{\phi},
\end{equation}
\begin{equation}\label{10.1.5}
e_{\theta}= \left(\frac{1}{r}+ O(r^{-4})\right)\partial_{\theta} + O(r^{-4})  \partial_{\phi},  \text{ }\text{ }\text{ }\text{ }\text{ }\text{ }\text{ }
e_{\phi} = \left(\frac{1}{r\sin{\theta}} +O(r^{-4}) \right) \partial_{\phi}.
\end{equation}
It now follows that
\begin{align} \label{10.1.6}
\begin{split}
&\lim_{r\rightarrow\infty}\int_{\overline{S}_{r}} \left( uQ(\nu_{g_{1}})-uk_{1}(\widetilde{v},\nu_{g_{1}})
-k_{1}(\widetilde{Y}^{\phi}\partial_{\phi}, \nu_{g_{1}}) \right)\\
=& \int_{S_{\infty}} u\sqrt{1+u^{2}|\nabla f|_{g}^{2}}(k-\pi)(w(e_{r})e_{r},e_{r})\\
=& 0,
\end{split}
\end{align}
since
\begin{equation}
u=O(r),\quad\quad (k-\pi)(e_{r}, e_{r})=O(r^{-2}),\quad\quad  w(e_{r})=O(r^{-2}).
\end{equation}

Next we treat the second integral on the right-hand side of \eqref{10.1.0}.
Observe that since
$k_{1}(\barna f, \nu_{g_{1}})= k_{1}(\barna \widetilde{f},\nu_{g_{1}}) =0$ on $\overline{S}_{r}$ it holds that
\begin{equation} \label{10.1.7}
\begin{split}
&\lim_{r\rightarrow 0} \int_{\overline{S}_{r}} \left( uQ(\nu_{g_{1}})-uk_{1}(\widetilde{v},\nu_{g_{1}})
-k_{1}(\widetilde{Y}^{\phi}\partial_{\phi}, \nu_{g_{1}})\right) \\
=&\lim_{r\rightarrow 0}\int_{\overline{S}_{r}} u(\operatorname{Hess}_{g_{1}} f)(Y, \nu_{g_{1}})
 + u(k-\pi)(\overline{w}, \nu_{g_{1}})
 + \frac{u^{2} \nu_{g_{1}}(f)(k-\pi)(\overline{w}, \overline{w})}{\sqrt{1-u^{2}|\barna f|_{g_{1}}^{2}}}
 - k_{1}(\widetilde{Y}^{\phi}\partial_{\phi}, \nu_{g_{1}}),
\end{split}
\end{equation}
where as given in \cite{ChaKhuri1}
\begin{align} \label{10.1.8}
\begin{split}
\pi_{ij} =& \frac{1}{\sqrt{1 - u^{2}|\barna f|_{g_{1}}^{2}}}
\left(  u\barna_{ij}f + (k_{1})_{ij} + \frac{f_{i}\psi_{j}+f_{j}\psi_{i}}{2u}
- \frac{ug_{1}^{lp}f_{l}\psi_{p}}{2}f_{i}f_{j}\right) \\
&+ \frac{1}{2}f_{i} \overline{w}^{l}(\overline{Y}_{l,j}-\overline{Y}_{j,l})
+ \frac{1}{2} f_{j}\overline{w}^{l}(\overline{Y}_{l,i}-\overline{Y}_{i,l}), \\
\overline{w} =& \frac{u \barna f + u^{-1}Y }{\sqrt{1-u^{2}|\barna f|_{g_{1}}^{2}}},\quad\quad\quad
\psi = u^{2} - Y^{\phi}\overline{Y}_{\phi},\quad\quad\quad \overline{Y}_{i} = (g_{1})_{i\phi}Y^{\phi}.
\end{split}
\end{align}
Direct computations show that
\begin{equation} \label{10.1.9}
\overline{w}^{r} = O(r^{3}), \text{ }\text{ }\text{ }
\overline{w}^{\theta} = O(r^{4}), \text{ }\text{ }\text{ }
\overline{w}^{\phi}  = O(r^{4}) \text{ }\text{ }\text{ in asymptotically AdS hyperbolic}\text{ }\text{ }(M_1)_{end}^{-},
\end{equation}
\begin{equation} \label{10.1.9b}
\overline{w}^{r} = O(r^{\frac{5}{2}}), \text{ }\text{ }\text{ }
\overline{w}^{\theta} = O(r^{\frac{3}{2}}), \text{ }\text{ }\text{ }
\overline{w}^{\phi}  = \frac{\mathcal{Y}}{r}+O(r^{-\frac{1}{2}})
\text{ }\text{ }\text{ in asymptotically cylindrical}\text{ }\text{ }(M_1)_{end}^{-},
\end{equation}
\begin{equation} \label{10.1.10}
\psi = r^{2} + O(r^{\frac{5}{2}}) \text{ }\text{ }\text{ }\text{ }\text{ in asymptotically AdS hyperbolic}\text{ }\text{ }(M_1)_{end}^{-},
\end{equation}
\begin{equation} \label{10.1.10b}
\psi = - \mathcal{Y}^{2}\sigma_{\phi \phi} + O(r^{\frac{1}{2}}) \text{ }\text{ }\text{ }\text{ }\text{ in asymptotically cylindrical}\text{ }\text{ }(M_1)_{end}^{-},
\end{equation}
\begin{equation} \label{10.1.11}
\begin{split}
&\overline{\Gamma}^{r}_{rr}= -r^{-1} + O(r),\quad
\overline{\Gamma}^{\alpha}_{rr}= O(r^{2}),\quad
\overline{\Gamma}^{\beta}_{r\alpha}= -r^{-1} \delta^{\beta}_{\alpha} + O(r),\quad
\overline{\Gamma}^{r}_{\alpha \beta}= r^{-1}\sigma_{\alpha \beta} + O(r), \\
&\overline{\Gamma}^{r}_{r\alpha}= O(r^{2}),\quad
\overline{\Gamma}^{\gamma}_{\alpha \beta}= (\Gamma_{\sigma})^{\gamma}_{\alpha \beta} + O(r^{2})
\text{ }\text{ in asymptotically AdS hyperbolic}\text{ }\text{ }(M_1)_{end}^{-},
\end{split}
\end{equation}
and
\begin{equation} \label{10.1.11b}
\begin{split}
&\overline{\Gamma}^{r}_{rr}= -r^{-1} + O(r^{-\frac{1}{2}}),\quad\quad
\overline{\Gamma}^{\alpha}_{rr}= O(r^{-\frac{3}{2}}),\quad\quad
\overline{\Gamma}^{\beta}_{r\alpha}= O(r^{-\frac{1}{2}}),\quad\quad
\overline{\Gamma}^{r}_{\alpha \beta}= O(r^{\frac{3}{2}}),\\
&\overline{\Gamma}^{r}_{r\alpha}= O(r^{\frac{1}{2}}),\quad\quad
\overline{\Gamma}^{\gamma}_{\alpha \beta}= (\Gamma_{\sigma})^{\gamma}_{\alpha \beta} + O(r^{\frac{1}{2}})
\text{ }\text{ in asymptotically cylindrical}\text{ }\text{ }(M_1)_{end}^{-},
\end{split}
\end{equation}
where $\overline{\Gamma}^{i}_{jl}$ are Christoffel symbols for $g_1$ and $(\Gamma_{\sigma})^{\gamma}_{\alpha \beta}$ are Christoffel symbols the round metric
$\sigma$ on $S^2$. From this, and the asymptotics in Section \ref{sec7} it follows that
\begin{equation} \label{10.1.12}
\pi_{rr}= O(1), \text{ }\text{ }\text{ }\text{ }\text{ }
\pi_{r \alpha} = O(r), \text{ }\text{ }\text{ }\text{ }\text{ }
\pi_{\alpha \beta}  = O(1) \text{ }\text{ in asymptotically AdS hyperbolic}\text{ }\text{ }(M_1)_{end}^{-},
\end{equation}
\begin{align} \label{10.1.12b}
\begin{split}
&\pi_{rr}= O(r^{-2}), \text{ }\text{ }\text{ }\text{ }
\pi_{r \theta} = O(r^{-\frac{3}{2}}), \text{ }\text{ }\text{ }\text{ }
\pi_{r \phi} = (k_{1})_{r \phi}+ O(1), \text{ }\text{ }\text{ }\text{ } \\
&\pi_{\theta \theta}  = O(r^{-\frac{1}{2}}),  \text{ }\text{ }\text{ }\text{ }
\pi_{\phi \alpha}  = O(1) \text{ }\text{ }\text{ }\text{ }
\text{ in asymptotically cylindrical}\text{ }\text{ }(M_1)_{end}^{-}.
\end{split}
\end{align}
Moreover \eqref{7.232}-\eqref{7.234} imply that
\begin{align} \label{10.1.13}
\begin{split}
&\nu_{g_{1}}= \left(r+ O(r^{3}) \right)\partial_{r} + O(r^{4})\partial_{\alpha}, \\
&dA_{\overline{S}_{r}} = \left(r^{-2} + O(r) \right) dA_{\sigma}
\text{ }\text{ }\text{ }\text{ }\text{ in asymptotically AdS hyperbolic}\text{ }\text{ }(M_1)_{end}^{-},
\end{split}
\end{align}
\begin{align} \label{10.1.13b}
\begin{split}
&\nu_{g_{1}}= \left(r+ O(r^{\frac{3}{2}}) \right)\partial_{r} + o(r^{\frac{1}{2}})\partial_{\alpha}, \\
&dA_{\overline{S}_{r}} = \left(1 + O(r^{\frac{1}{2}}) \right) dA_{\sigma}
\text{ }\text{ }\text{ }\text{ }\text{ in asymptotically cylindrical}\text{ }\text{ }(M_1)_{end}^{-},
\end{split}
\end{align}
where $dA_{\overline{S}_{r}}$ and $dA_{\sigma}$ are respective area forms.
The above calculations together with \eqref{7.11}, \eqref{7.15}, and \eqref{7.22} produce
\begin{equation} \label{10.1.14}
u(\operatorname{Hess}_{g_{1}} f)(Y, \nu_{g_{1}}) = O(r^{9}), \text{ }\text{ }\text{ }\text{ }\text{ }\text{ }\text{ }
u(k-\pi)(\overline{w}, \nu_{g_{1}}) = O(r^{5}),
\end{equation}
\begin{equation} \label{10.1.15}
\frac{u^{2} \nu_{g_{1}}(f)(k-\pi)(\overline{w}, \overline{w})}{\sqrt{1-u^{2}|\barna f|_{g_{1}}^{2}}} = O(r^{9})
\text{ }\text{ }\text{ in asymptotically AdS hyperbolic}\text{ }\text{ }(M_1)_{end}^{-}.
\end{equation}
The desired result now follows from \eqref{10.1.13}, \eqref{10.1.14}, and \eqref{10.1.15} in the case that the secondary end is asymptotically AdS hyperbolic; note that the last term in \eqref{10.1.7} vanishes due to the asymptotics \eqref{7.361}.

Consider now the case when the secondary end is asymptotically cylindrical. Similar calculations to those above yield
\begin{equation} \label{10.1.14b}
u(\operatorname{Hess}_{g_{1}} f)(Y, \nu_{g_{1}}) = O(r^{2}), \quad\text{ }\text{ }\text{ }
u(k-\pi)(\overline{w}, \nu_{g_{1}}) = (k-k_{1})(Y, \nu_{g_{1}}) + O(r^{\frac{1}{2}}),
\end{equation}
\begin{equation} \label{10.1.15b}
\frac{u^{2} \nu_{g_{1}}(f)(k-\pi)(\overline{w}, \overline{w})}{\sqrt{1-u^{2}|\barna f|_{g_{1}}^{2}}} = O(r^{\frac{1}{2}})
\text{ }\text{ }\text{ in asymptotically cylindrical}\text{ }\text{ }(M_1)_{end}^{-}.
\end{equation}
Together with \eqref{10.1.13b} this shows that \eqref{10.1.7} becomes
\begin{equation} \label{10.1.16}
\lim_{r\rightarrow 0} \int_{\overline{S}_{r}} \left( uQ(\nu_{g_{1}})-uk_{1}(\widetilde{v},\nu_{g_{1}}) \right)
= \lim_{r\rightarrow 0}\int_{\overline{S}_{r}} (k-k_{1})(Y, \nu_{g_{1}})
=  \lim_{r\rightarrow 0}\int_{\overline{S}_{r}} k (Y, \nu_{g_{1}}) - 8\pi \mathcal{Y} \mathcal{J}_{1}.
\end{equation}
Analogous computations to (9.21)-(9.29) in \cite{ChaKhuri1} give
\begin{equation} \label{10.1.17}
k(Y, \nu_{g}) = (1+O(r^{3})) k (Y, \nu_{g_{1}}) + O(r^{3})k (Y, \partial_{\theta}) + O(r^{\frac{1}{2}})k(Y, \partial_{\phi}),
\text{ }\text{ }\text{ }\text{ }
dA_{S_{r}} = (1+O(r^{3}))dA_{\overline{S}_{r}},
\end{equation}
and therefore
\begin{equation}\label{10.1.18}
\lim_{r\rightarrow 0} \int_{\overline{S}_{r}} \left( uQ(\nu_{g_{1}})-uk_{1}(\widetilde{v},\nu_{g_{1}}) \right)
= 8\pi \mathcal{Y} (\mathcal{J}-\mathcal{J}_{1}),
\end{equation}
which vanishes if and only if $\mathcal{J} = \mathcal{J}_{1}$.  The last term in
\eqref{10.1.0} immediately vanishes due to the asymptotics \eqref{7.362}.
\end{proof}

\begin{lemma}\label{lemma9.3}
Under the hypotheses and notation of Theorems \ref{thm7.4} and \ref{thm8.1}
\begin{equation}
\int_{M_{3}}\operatorname{div}_{g_{3}}(\widetilde{u}\widetilde{Q}) =0.
\end{equation}
\end{lemma}

\begin{proof}
The divergence theorem gives
\begin{equation} \label{10.3.0}
\int_{M_{3}}\operatorname{div}_{g_{3}}(\widetilde{u}\widetilde{Q}) \\
=  \lim_{r\rightarrow\infty}\int_{\widetilde{S}_{r}}  \widetilde{u}\widetilde{Q}(\nu_{g_{3}})
- \lim_{r\rightarrow 0}\int_{\widetilde{S}_{r}} \widetilde{u}\widetilde{Q}(\nu_{g_{3}}).
\end{equation}
It is shown in \cite{ChaKhuriSakovich} that the first integral vanishes, and that the second integral also vanishes when the secondary end is asymptotically cylindrical \cite{ChaKhuri1}. For this latter result the following condition is needed
\begin{equation} \label{10.3.0}
|k_{1}+g_{1}|_{g_{1}} + |(k_{1}+g_{1})(\partial_{\phi}, \cdot)|_{g_{1}} + |(k_{1}+g_{1})(\partial_{\phi}, \partial_{\phi})|
=O(1),
\end{equation}
which is straightforward to derive from \eqref{7.1b}, \eqref{7.12}, \eqref{7.16}, and \eqref{7.23}. Hence, we will focus on the second integral when $(M_{3})^{-}_{end}$ is asymptotically flat.

Similar computations to those in Appendix B of \cite{ChaKhuriSakovich} apply here to produce
\begin{equation} \label{10.3.1}
\lim_{r\rightarrow 0} \int_{\widetilde{S}_{r}} \widetilde{u}\widetilde{Q}(\nu_{g_{3}})
=\lim_{r\rightarrow 0} \int_{\widetilde{S}_{r}} \widetilde{u}(\operatorname{Hess}_{g_{3}} \widetilde{f})(\widetilde{Y}, \nu_{g_{3}})
 + \lim_{r\rightarrow 0} \int_{\overline{S}_{r}} \widetilde{u}\sqrt{1+\widetilde{u}^{2}|\widetilde{\nabla} \widetilde{f}|_{g_{2}}^{2}}(k_{2}-\pi_{2})(\widetilde{w},\nu_{g_{2}}),
\end{equation}
where
\begin{equation}\label{10.3.2}
(\pi_{2})_{ij}=\frac{ \widetilde{u} \widetilde{\nabla}_{ij}f
+\widetilde{u}_{i}\widetilde{f}_{j}
+\widetilde{u}_{j}\widetilde{f}_{i}
+\frac{1}{2\widetilde{u}}\left((g_{2})_{i\phi}\widetilde{Y}^{\phi}_{,j}+(g_{2})_{j\phi}
\widetilde{Y}^{\phi}_{,i}\right)}{\sqrt{1+\widetilde{u}^{2}|\widetilde{\nabla} \widetilde{f}|_{g_{2}}^{2}}},\text{ }\text{ }\text{ }\text{ }\text{ }\text{ }\text{ }
\widetilde{w}^{i}=\frac{\widetilde{u}\widetilde{f}^{i} +\widetilde{u}^{-1}\widetilde{Y}^{i}}{\sqrt{1+\widetilde{u}^{2}|\widetilde{\nabla} \widetilde{f}|_{g_{2}}^{2}}}.
\end{equation}
Asymptotics for the Christoffel symbols of $g_3$ are given by
\begin{equation} \label{10.3.3}
\widetilde{\Gamma}^{r}_{r \phi}= o(r^{\frac{1}{2}}),\text{ }\text{ }\text{ }\text{ }\text{ }
\widetilde{\Gamma}^{\theta}_{r \phi}= o(r^{-\frac{1}{2}}),\text{ }\text{ }\text{ }\text{ }\text{ }
\widetilde{\Gamma}^{r}_{\alpha \phi}= r \sigma_{\alpha \phi}+o(r^{\frac{3}{2}}),\text{ }\text{ }\text{ }\text{ }\text{ }
\widetilde{\Gamma}^{\theta}_{\alpha \phi}= -\frac{1}{2}\sigma_{\alpha \phi, \theta}  + o(r^{\frac{1}{2}}),
\end{equation}
and
\begin{equation} \label{10.3.4}
\nu_{g_{3}}= \left(r^{2}+ o(r^{5/2}) \right)\partial_{r} + o(r^{\frac{3}{2}})\partial_{\alpha}, \text{ }\text{ }\text{ }\text{ }\text{ }
dA_{\widetilde{S}_{r}} = \left(r^{-2} + o(r^{-\frac{3}{2}}) \right) dA_{\sigma}.
\end{equation}
It follows that
\begin{equation} \label{10.3.5}
\widetilde{u}(\operatorname{Hess}_{g_{3}} \widetilde{f})(\widetilde{Y}, \nu_{g_{3}}) = o(r^{\frac{11}{2}}),
\end{equation}
and therefore the first integral in \eqref{10.3.1} vanishes.

Consider now the second integral in \eqref{10.3.1}. Observe that \eqref{7.361}, \eqref{7.420}, and \eqref{7.46} imply
\begin{align} \label{10.3.6}
\begin{split}
&(\pi_{2})_{rr}= r^{-2} + 2\left( \mathcal{C}_{4}-\mathcal{D}_{3} \right)r^{-1} + o(r^{-\frac{1}{2}}), \text{ }\text{ }\text{ }\text{ }\text{ } \\
&(\pi_{2})_{r \alpha} = \left(\mathcal{C}_{4} + \mathcal{D}_{3} \right)_{,\alpha} + o(r^{\frac{1}{2}}), \text{ }\text{ }\text{ }\text{ }\text{ } \\
&(\pi_2)_{\alpha \beta}  = -r^{-2} \sigma_{\alpha \beta} + o(r^{-\frac{1}{2}}),
\end{split}
\end{align}
%\textcolor{blue}{Note that the first and the third will start with the terms with opposite signs %upto the sign of the first term in the asymptotics of $f$. This yields that $k_{2}$ cannot be %equal to $\pi_{2}$. This gives another reason why the rigidity case cannot be realized by our %method.}
and we also have
\begin{equation} \label{10.3.7}
\sqrt{1+\widetilde{u}^{2}|\widetilde{\nabla} \widetilde{f}|_{g_{2}}^{2}} \widetilde{w}^{r}
= 1 + \left( \mathcal{C}_{4} - 2\mathcal{D}_{3} \right)r + o(r^{\frac{3}{2}}),
\quad
\sqrt{1+\widetilde{u}^{2}|\widetilde{\nabla} \widetilde{f}|_{g_{2}}^{2}} \widetilde{w}^{\alpha}
= r^{2}\sigma^{\alpha \beta} \left(\mathcal{C}_{4} \right)_{,\beta} + o(r^{\frac{5}{2}}),
\end{equation}
where $\mathcal{D}_3$ is the coefficient of the $r^{-2}$-term in the expansion \eqref{7.46}.
It then follows from \eqref{7.420} and \eqref{10.1.13} that
\begin{equation} \label{10.3.8}
\lim_{r\rightarrow 0} \int_{\overline{S}_{r}} \widetilde{u}\sqrt{1+\widetilde{u}^{2}|\widetilde{\nabla} \widetilde{f}|_{g_{2}}^{2}}(k_{2}-\pi_{2})(\widetilde{w},\nu_{g_{2}})
=  \int_{S^{2}} 2 \left( \mathcal{D}_{3}- \mathcal{C}_{4} \right) + \lim_{r\rightarrow 0} \int_{\overline{S}_{r}} \nu_{g_{2}}^{\alpha} \left( \mathcal{D}_{3}-\mathcal{C}_{4}\right)_{,\alpha}.
\end{equation}
Since $\mathcal{D}_{3}=\mathcal{C}_4$ the desired result is attained.
\end{proof}

\begin{lemma}\label{lemma10.4}
\begin{equation} \label{10.2}
J_{1}(\widetilde{v})
= \left(\operatorname{div}_{g_{1}} k_{1}\right)(\widetilde{v})
= u^{-1} \operatorname{div}_{g_{1}} \left(u k_{1}(\widetilde{v}, \cdot)  \right)
-g_{1} \left( k_{1}, \pi_{2} \right)
+ \frac{ \operatorname{div}_{g_{1}} \left( k_{1}( \widetilde{Y}^{\phi} \partial_{\phi}, \cdot)  \right)}{\widetilde{u}\sqrt{1+ \widetilde{u}^{2}|\widetilde{\nabla} \widetilde{f}|_{g_{2}}^{2}}}
\end{equation}
\end{lemma}

\begin{proof}
From \eqref{7.24} we have
\begin{equation} \label{10.2.1}
k_{1}^{jp} = g_{1}^{ij}g_{1}^{lp}\left(\frac{(g_{1})_{l \phi} \partial_{i}Y^{\phi}+(g_{1})_{i \phi} \partial_{l}Y^{\phi}}{2u}\right)
= \frac{\delta^p_{\phi} g_{1}^{ij}\partial_{i}Y^{\phi}+\delta^j_{\phi} g_{1}^{lp}\partial_{l}Y^{\phi}}{2u}.
\end{equation}
Therefore, for any axially symmetric functions $h_{1}$ and $h_{2}$ it holds that
\begin{equation} \label{10.2.2}
k_{1}^{jp}\partial_{j}h_{1}\partial_{p}h_{2}=0.
\end{equation}
This eventually yields
\begin{equation} \label{10.2.3}
\begin{split}
\left(\operatorname{div}_{g_{1}} k_{1}\right)(\widetilde{v})
&= \operatorname{div}_{g_{1}} \left( k_{1}(\widetilde{v}, \cdot)  \right)
- k_{1}^{jp} \overline{\nabla}_{j} \widetilde{v}_p   \\
&= u^{-1}\operatorname{div}_{g_{1}} \left( k_{1}(u\widetilde{v}, \cdot)  \right)
- g_{1} \left( k_{1}, \pi_{2} \right)
+ \frac{k_{1} (g_{1}^{lp} \partial_{l}\widetilde{Y}^{\phi}\partial_{p}, \partial_{\phi})}{\widetilde{u}\sqrt{1+ \widetilde{u}^{2}|\widetilde{\nabla} \widetilde{f}|_{g_{2}}^{2}}} \\
&= u^{-1}\operatorname{div}_{g_{1}} \left( k_{1}(u\widetilde{v}, \cdot)  \right)
- g_{1} \left( k_{1}, \pi_{2} \right)
+ \frac{ \operatorname{div}_{g_{1}} \left( k_{1}( \widetilde{Y}^{\phi} \partial_{\phi}, \cdot)  \right)}{\widetilde{u}\sqrt{1+ \widetilde{u}^{2}|\widetilde{\nabla} \widetilde{f}|_{g_{2}}^{2}}}, \\
\end{split}
\end{equation}
since
\begin{equation} \label{10.2.4}
\overline{\nabla}_{j} \widetilde{v}_p
= (\pi_{2})_{jp}
- \frac{\widetilde{u}_{p}\widetilde{f}_{j}
+ \frac{1}{2\widetilde{u}}\left((g_{1})_{j\phi}\widetilde{Y}^{\phi}_{,p}+(g_{1})_{p\phi}
\widetilde{Y}^{\phi}_{,j}\right)}{\sqrt{1+\widetilde{u}^{2}|\widetilde{\nabla} \widetilde{f}|_{g_{2}}^{2}}}
+ \widetilde{u} \widetilde{f}_{p} \partial_{j} \left( 1+\widetilde{u}^{2}|\widetilde{\nabla} \widetilde{f}|_{g_{2}}^{2} \right)^{-\frac{1}{2}}.
\end{equation}
\end{proof}

\section{Appendix B: The Kerr-AdS Black Hole}

\label{set14} \setcounter{equation}{0}
\setcounter{section}{13}

Here we discuss the Kerr-AdS metric with the aim of motivating the asymptotics
for $Y^{\phi}$ in Section \ref{sec7}. Recall that the domain of outer communication of the Kerr-AdS spacetime has topology $\mathbb{R}\times [\widetilde{r}_{+},\infty)\times S^2$, and the metric with cosmological constant $\Lambda = -3$ in Boyer-Lindquist coordinates is given by
\begin{equation} \label{14.1}
\begin{split}
\widetilde{g} =& - \left(1 - \frac{2\mathfrak{m} \widetilde{r}}{\widetilde{r}^{2}+a^{2}\cos^{2}{\widetilde{\theta}}} + \widetilde{r}^{2}+a^{2}\sin^{2}{\widetilde{\theta}}\right)d\widetilde{t}^{2} \\
 &- \frac{2a \sin^{2}{\widetilde{\theta}}}{1-a^{2}}\left( \frac{2\mathfrak{m}\widetilde{r}}{\widetilde{r}^{2}+a^{2}\cos^{2}{\widetilde{\theta}}} -(\widetilde{r}^{2}+a^{2}) \right) d\widetilde{t} d\widetilde{\phi}\\
 &+ \frac{\widetilde{r}^{2}+a^{2}\cos^{2}{\widetilde{\theta}}}{\widetilde{r}^{4}+(a^{2}+1)
 \widetilde{r}^{2}-2\mathfrak{m}\widetilde{r}+a^{2}} d\widetilde{r}^{2}
 + \frac{\widetilde{r}^{2}+a^{2}\cos^{2}{\widetilde{\theta}}}{1-a^{2}\cos^{2}
 {\widetilde{\theta}}}d\widetilde{\theta}^{2}  \\
 &+ \frac{\sin^{2}{\widetilde{\theta}}}{(1-a^{2})^{2}} \left( \frac{2\mathfrak{m}\widetilde{r}a^{2}\sin^{2}{\widetilde{\theta}}}
 {\widetilde{r}^{2}+a^{2}\cos^{2}{\widetilde{\theta}}}
  + (\widetilde{r}^{2}+a^{2})(1-a^{2})\right) d\widetilde{\phi}^{2},
\end{split}
\end{equation}
where $\mathfrak{m}$ and $a$ are parameters related to the mass and angular momentum through the formulas
\begin{equation} \label{14.7}
m = \frac{\mathfrak{m}}{(1-a^{2})^{2}},\quad\quad\quad \mathcal{J}= -\frac{\mathfrak{m}a}{(1-a^{2})^{2}}.
\end{equation}
The solution is valid for $|a|<1$ and is singular otherwise. The event horizon is located at the
larger of the two positive roots $\widetilde{r}_{+}\geq \widetilde{r}_{-}$ of the polynomial
$\Delta_{\widetilde{r}}=
\widetilde{r}^{4}+(a^{2}+1)\widetilde{r}^{2}-2\mathfrak{m}\widetilde{r}+a^{2}$. It can be shown \cite{CaldarelliKlemm} that
\begin{equation} \label{14.2}
\mathfrak{m} \geq \frac{1}{3\sqrt{6}}\left(\sqrt{1+14a^{2}+a^{4}} + 2a^{2} +2 \right)
\left(\sqrt{1+14a^{2}+a^{4}} - a^{2} -1 \right)^{\frac{1}{2}},
\end{equation}
unless the metric contains a naked singularity. Equality is achieved in \eqref{14.2} in the case of an extreme black hole, that is when $\widetilde{r}_{+}= \widetilde{r}_{-}$.

As pointed out in \cite{HenneauxTeitelboim}, the line element of \eqref{14.1} does not satisfy the asymptotics in \eqref{2}. However, this may be remedied by an appropriate coordinate transformation (cf. Appendix B of \cite{HenneauxTeitelboim})
\begin{equation} \label{14.3}
t = \widetilde{t}, \quad\quad
\phi = \widetilde{\phi} + a \widetilde{t}, \quad\quad
r \cos{\theta} = \widetilde{r} \cos{\widetilde{\theta}}, \quad\quad
(1-a^{2})r^{2} = \widetilde{r}^{2} (1-a^{2}\cos^{2}{\widetilde{\theta}}) + a^{2}\sin^{2}{\widetilde{\theta}}.
\end{equation}
In particular, under this new coordinate system the components of the metric have the
following asymptotics
\begin{align} \label{14.4}
\begin{split}
&\widetilde{g}_{tt} = -(1+r^{2}) + \frac{2\mathfrak{m}}{r} (1-a^{2}\sin^{2}{\theta})^{-\frac{5}{2}} + O(r^{-3}), \\
&\widetilde{g}_{t \phi} = -\frac{2\mathfrak{m}a \sin^{2}{\theta}}{r} (1-a^{2}\sin^{2}{\theta})^{-\frac{5}{2}} +O(r^{-3}), \\
&\widetilde{g}_{\phi \phi} = (g_{0})_{\phi \phi} + \frac{2\mathfrak{m}a^{2}\sin^{4}{\theta}}{r}(1-a^{2}\sin^{2}{\theta})^{-\frac{5}{2}} +O(r^{-3}), \\
&\widetilde{g}_{rr} = (g_{0})_{rr} + \frac{2\mathfrak{m}}{r^{5}}(1-a^{2}\sin^{2}{\theta})^{-\frac{3}{2}} + O(r^{-7}), \\
&\widetilde{g}_{r\theta} = -\frac{2\mathfrak{m}a^{2}\sin{\theta}\cos{\theta}}{r^{4}}(1-a^{2}\sin^{2}{\theta})^{-\frac{5}{2}} +O(r^{-6}), \\
&\widetilde{g}_{\theta \theta} = (g_{0})_{\theta \theta} + \frac{2\mathfrak{m}a^{4}\sin^{2}{\theta}\cos^{2}{\theta}}{r^{3}}
(1-a^{2}\sin^{2}{\theta})^{-\frac{7}{2}} +O(r^{-5}),
\end{split}
\end{align}
where $g_{0}$ is the hyperbolic metric as given in Section \ref{sec2}. It is then clear that the induced metric on the $t=0$ slice satisfies \eqref{2}.

According to the definition of $Y^{\phi}$ we find that
\begin{equation} \label{14.5}
Y^{\phi} = -(\widetilde{g}_{\phi \phi})^{-1} \widetilde{g}_{t \phi} = \frac{2\mathfrak{m}a}{r^{3}}(1-a^{2}\sin^{2}{\theta})^{-\frac{5}{2}} +O(r^{-4}).
\end{equation}
This motivates the asymptotics for $Y^{\phi}$ in \eqref{7.10}, in which $\mathcal{J}(\theta)=- \mathfrak{m}a(1-a^{2}\sin^{2}{\theta})^{-\frac{5}{2}}$. It is straightforward to check that the total angular momentum is given by
\begin{equation} \label{14.6}
\mathcal{J} = \frac{3}{4} \int_{0}^{\pi} \mathcal{J}(\theta) \sin^{3}{\theta} d\theta.
\end{equation}

\end{document}